\documentclass[acmsmall,screen,nonacm]{acmart}

\usepackage{amsmath,amsthm,amsbsy}
\usepackage{graphicx}
\usepackage{caption,wrapfig,picins}
\usepackage{subcaption}
\usepackage{longtable}
\usepackage{booktabs}
\usepackage{tabu}
\usepackage{tabularx}
\usepackage[referable]{threeparttablex}
\usepackage{multirow}
\usepackage{subcaption}
\usepackage{hyperref}
\usepackage[capitalise]{cleveref}
\usepackage{stmaryrd}
\usepackage{mathtools} 
\usepackage{makecell}
\crefname{figure}{Figure}{Figure}
\crefformat{footnote}{#2\footnotemark[#1]#3}
\usepackage{tikz}
\usepackage{comment}
\usepackage{parskip}
\usepackage{paralist}
\usepackage{pgfplots}

\usepackage[ruled, linesnumbered, noend]{algorithm2e}
\usepackage{tikz}
\usetikzlibrary{shapes.misc,arrows.meta}
\usetikzlibrary{calc,decorations.pathmorphing,shapes}
\usetikzlibrary{decorations.markings}
\usetikzlibrary{decorations.pathreplacing}
\tikzset{snake it/.style={decorate, decoration=snake}}

\usepackage{xspace}
\usepackage{xcolor}
\usepackage[inline, shortlabels]{enumitem}

\usepackage{multicol}
\usepackage{bbm}

\usepackage{thm-restate}

\newcommand{\figlabel}[1]{\label{fig:#1}}
\newcommand{\figref}[1]{Fig.~\ref{fig:#1}}
\newcommand{\tablabel}[1]{\label{tab:#1}}
\newcommand{\tabref}[1]{Table~\ref{tab:#1}}
\newcommand{\seclabel}[1]{\label{sec:#1}}
\newcommand{\secref}[1]{Section~\ref{sec:#1}}

\newcommand{\deflabel}[1]{\label{def:#1}}
\newcommand{\defref}[1]{Definition~\ref{def:#1}}

\newcommand{\thmlabel}[1]{\label{thm:#1}}
\newcommand{\thmref}[1]{Theorem~\ref{thm:#1}}
\newcommand{\proplabel}[1]{\label{prop:#1}}
\newcommand{\propref}[1]{Proposition~\ref{prop:#1}}

\newcommand{\applabel}[1]{\label{app:#1}}
\newcommand{\appref}[1]{Appendix~\ref{app:#1}}

\makeatletter
\def\thm@space@setup{
\thm@postskip=0pt}
\makeatother

\theoremstyle{definition}
\newtheorem{example}{Example}
\numberwithin{example}{section}

\newtheorem{theorem}{Theorem}
\numberwithin{theorem}{section}

\newtheorem{definition}{Definition}
\numberwithin{definition}{section}

\numberwithin{problem}{section}

\newtheorem{lemma}{Lemma}
\numberwithin{lemma}{section}

\newtheorem{proposition}{Proposition}
\numberwithin{proposition}{section}

\numberwithin{corollary}{section}

\numberwithin{claim}{section}

\newtheorem{remark}{Remark}

\newenvironment{myparagraph}[1]{\vspace{0in}\noindent{\bf #1.}}{}

\DeclareMathAlphabet{\mathpzc}{OT1}{pzc}{m}{it}

\newcommand{\set}[1]{\{#1\}}
\newcommand{\setpred}[2]{\{#1 \,|\, #2\}}
\newcommand{\proj}[2]{#1|_{#2}}
\renewcommand{\emptyset}{\varnothing}
\newcommand{\tuple}[1]{\langle #1 \rangle}
\newcommand{\concat}{}

\newcommand{\ev}[1]{\langle #1\rangle}
\newcommand{\mems}{\mathcal{X}}
\newcommand{\locks}{\mathcal{L}}
\newcommand{\threads}{\mathcal{T}}
\newcommand{\vals}{\mathcal{V}}
\newcommand{\ThreadOf}[1]{\mathsf{thr}(#1)}
\newcommand{\OpOf}[1]{\mathsf{op}(#1)}
\newcommand{\MemOf}[1]{\mathsf{mem}(#1)}
\newcommand{\LockOf}[1]{\mathsf{lock}(#1)}
\newcommand{\ObjOf}[1]{\mathsf{obj}(#1)}

\newcommand{\events}[1]{\mathsf{Events}_{#1}}
\newcommand{\tr}{\sigma}
\newcommand{\rf}[1]{\mathsf{rf}_{#1}}
\newcommand{\po}[1]{\mathsf{po}_{#1}}

\newcommand{\creorderings}[1]{\textsf{CReorderings}(#1)}

\newcommand{\eqcl}[2]{[#1]_{#2}}

\newcommand{\maz}{\mathcal{M}}
\newcommand{\mazeq}{\equiv_{\maz}}
\newcommand{\mazcl}[1]{\eqcl{#1}{\maz}}

\newcommand{\indrel}{\mathbb{I}}

\newcommand{\labs}{\Sigma}
\newcommand{\mem}{\textsf{mem}}
\newcommand{\lck}{\textsf{lck}}

\newcommand{\rfcl}[1]{\eqcl{#1}{\rf{}}}

\newcommand{\subsq}[1]{\mathbf{i}}

\newcommand{\grains}{\mathcal{G}} 

\newcommand{\graincl}[1]{\eqcl{#1}{\grains}} 

\newcommand{\scc}{\mathcal{C}}
\newcommand{\scatteredgrains}{\mathcal{SG}}
\newcommand{\scgraincl}[1]{\eqcl{#1}{\scatteredgrains}} 

\newcommand{\opfont}[1]{\texttt{#1}}
\newcommand{\acq}{\opfont{acq}}
\newcommand{\rel}{\opfont{rel}}
\newcommand{\rd}{\opfont{r}}
\newcommand{\wt}{\opfont{w}}

\newcommand{\lk}{\ell}

\newcommand{\poly}{\operatorname{poly}}

\newcommand{\acr}[1]{{\sf #1}}
\newcommand{\acrtr}{\operatorname{\acr{seq}}}
\newcommand{\ord}[2]{\leq^{#1}_{\mathsf{#2}}}

\newcommand{\strictord}[2]{<^{#1}_{\mathsf{#2}}}

\newcommand{\trord}[1]{\ord{#1}{\acrtr}}
\newcommand{\stricttrord}[1]{\strictord{#1}{\acrtr}}

\newcommand{\aut}{\mathcal{A}}

\newcommand{\chk}{{\sf check}}

\newcommand{\subsumes}{\preceq}

\newcommand{\syncp}{\acr{SyncP}\xspace}
\newcommand{\confp}{\acr{ConfP}\xspace}
\newcommand{\seqp}{\acr{SeqP}\xspace}
\newcommand{\grconfp}{\grains\textsf{Prefix}\xspace}
\newcommand{\optseqp}{{\sf Opt}\acr{SeqP}\xspace}
\newcommand{\optgrconfp}{{\sf Opt}\grains\textsf{Prefix}\xspace}
\newcommand{\hb}{\acr{HB}\xspace}
\newcommand{\shb}{\acr{SHB}\xspace}
\newcommand{\osr}{\acr{OSR}\xspace}
\newcommand{\cp}{\acr{CP}\xspace}
\newcommand{\wcp}{\acr{WCP}\xspace}
\newcommand{\sdp}{\acr{SDP}\xspace}
\newcommand{\mtwo}{\acr{M2}\xspace}
\newcommand{\aftset}{\acr{AftSet}\xspace}

\colorlet{RED}{red}

\newcommand{\racy}[1]{\textcolor{red}{#1}}

\newcommand{\execution}[2]{
\scalebox{0.7}{
  \begin{tikzpicture}%
    \foreach \x in {1,...,#1}
    \node[right] at (1.5*\x+0.2,0.25) {$T_{\x}$};
    \draw (1.2,0) -- (#1*1.5+1.2,0);%
    \pgfmathsetmacro{\y}{1};%
    #2%
    \draw (1.2,0) -- (1.2,-0.4*\y);%
    \draw (#1*1.5+1.2,0) -- (#1*1.5+1.2,-0.4*\y);%
    \foreach \x in {2,...,#1}
    \draw[dashed] (1.5*\x-0.3,0) -- (1.5*\x-0.3,-0.4*\y);%
    \draw (1.2,-0.4*\y) -- (#1*1.5+1.2,-0.4*\y);%
  \end{tikzpicture}
}
}

\newcommand{\figev}[2]{
\pgfmathsetmacro{\y}{\y+1};
\pgfmathsetmacro{\y}{\y-1};
\node [left] at (1.25,-0.4*\y)  {\pgfmathprintnumber{\y}};%
\node at (#1*1.5 + 0.45,-0.4*\y) { #2 };%
\pgfmathsetmacro{\y}{\y+1};
}

\newcommand{\executionfull}[9]{
\scalebox{#3}{
  \begin{tikzpicture}
    \foreach \x in {1,...,#1}
    \node[right] at (#4*\x+#6, #8) {$T_{\x}$}; 
    \draw (#7,0) -- (#1*#4+#7,0); 
    \pgfmathsetmacro{\y}{1};
    #2 
    \draw (#7,0) -- (#7,-#5*\y); 
    \draw (#1*#4+#7,0) -- (#1*#4+#7,-#5*\y); 
    \draw (#7,-#5*\y) -- (#1*#4+#7,-#5*\y); 
    \ifthenelse{#9 = 1}{
      \foreach \x in {2,...,#1}
      \draw[dashed] (#4*\x+#7-#4,0) -- (#4*\x+#7-#4,-#5*\y); 
    }{}
  \end{tikzpicture}
}
}

\newcommand{\figevfull}[9]{
\ifthenelse{#7 = 1}{
  \ifthenelse{#8 = -1}{
    \node [left] at (#5,(-#4*\y))  {\pgfmathprintnumber{\y}};%
  }{
    \node [left] at (#5,(-#4*\y))  {#9};%
  }
}{}
\node at (#1*#3 + #6,(-#4*\y)) {$ #2 $};%
\pgfmathsetmacro{\y}{\y+1};
}

\newcommand{\trgrconfpwins}{\tr_{7}}
\newcommand{\trgrconfpwinsreordering}{\rho_{7}}
\newcommand{\trsyncpnotcomm}{\tr_{8}}
\newcommand{\trsyncpnotcommreordering}{\rho_{8}}
\newcommand{\trosrmisses}{\tr_{9}}
\newcommand{\trosrmissesreordering}{\rho_{9}}

\newcommand{\hLRU}{{\sf LRU}\xspace}
\newcommand{\hSz}{{\sf Sz}\xspace}
\newcommand{\hSh}{{\sf Sh}\xspace}

\newtoggle{appendix}
\toggletrue{appendix}

\begin{document}

\title{Enhanced Data Race Prediction Through Modular Reasoning}

\author{Zhendong Ang}
\email{zhendong.ang@u.nus.edu}
\affiliation{%
  \institution{National University of Singapore}
  \city{Singapore}
  \country{Singapore}
}

\author{Azadeh Farzan}
\email{azadeh@cs.toronto.edu}
\affiliation{%
  \institution{University of Toronto}
  \city{Toronto}
  \country{Canada}
}

\author{Umang Mathur}
\email{umathur@comp.nus.edu.sg}
\affiliation{%
  \institution{National University of Singapore}
  \city{Singapore}
  \country{Singapore}
}



\begin{abstract}
There are two orthogonal methodologies for efficient prediction of data races from concurrent program runs: commutativity and prefix reasoning. There are several instances of each methodology in the literature, with the goal of predicting data races using a streaming algorithm where the required memory does not grow proportional to the length of the observed run, but these instances were mostly created in an ad hoc manner, without much attention to their unifying underlying principles. 
In this paper, we identify and formalize these principles for each category with the ultimate goal of paving the way for combining them into a new algorithm which shares their efficiency characteristics but offers strictly more prediction power. In particular, we formalize three distinct classes of races predictable using commutativity reasoning, and compare them. We identify three different styles of prefix reasoning, and prove that they predict the same class of races, which provably contains all races predictable by any commutativity reasoning technique. 

Our key contribution is combining prefix reasoning and commutativity reasoning in a modular way to introduce a new class of races, {\em granular prefix races}, that are predictable in constant-space and linear time, in a streaming fashion. This class of races includes all races predictable using commutativity and prefix reasoning techniques. We present an improved constant-space algorithm for prefix reasoning alone based on the idea of antichains (from language theory). This improved algorithm is the stepping stone that is required to devise an efficient algorithm for prediction of granular prefix races. We present experimental results to demonstrate the expressive power and performance of our new algorithm.
\end{abstract}



\maketitle


\section{Introduction}
\seclabel{intro}

Dynamic data race detectors have emerged as the first line of defense against data races, which are often symptomatic of deeper and critical problems in concurrent software, and yet inherently hard to find.
But the effectiveness of such tools can be sensitive to
thread scheduling observed during the execution of the underlying program-under-test,
and is often poor. 
{\em Predictive} data race detection techniques attempt to identify races by looking at all alternate executions of the underlying program that can be inferred from the observed execution.
In its full generality, predictive data race detection is an intractable problem~\cite{Mathur2020b}, 
thanks to the exponentially many reorderings that need to be enumerated to expose a data race.
Nevertheless, sound polynomial time algorithms have emerged 
recently~\cite{Smaragdakis12,Kini17,genc2019,Roemer20,Mathur18,Pavlogiannis2020,OSR2024,Mathur21,seqcheck2021,Ang2024CAV} that trade completeness for running time,
often achieving the holy grail of runtime verification --- a monitorable
implementation, i.e., a streaming algorithm whose
memory requirement does not increase with the length of the execution.

This work was motivated by the questions 
(1) what are the key ingredients behind fast highly predictive 
data race detection algorithms?, and 
(2) can they be combined to yield better algorithms?
A clean and elegant answer to (1) has been elusive so far, and likely for that reason, (2) has never been systematically studied before. 
Existing algorithms, and their correctness proofs, have been
difficult~\cite{Smaragdakis12,Kini17}, largely unprincipled, 
and sometimes even incorrect~\cite{Roemer18,Kini17,Mathur18}.
In this work, we take on the challenge of demystifying the ingredients
behind fast algorithms, and identify two key principles of reasoning that
allow for monitorability  --- \emph{commutativity} and \emph{prefix reasoning}.

Commutativity reasoning is older and is based on a {\em sound} commutativity relation over individual events in the observed program run. 
The race prediction question boils to checking if two conflicting events can be made adjacent through a sequence of valid swaps between commuting events. 
Data race prediction techniques based on the happens-before partial order~\cite{Flanagan09,Pozniansky03,djit1999,Elmas07} and its variants~\cite{Mathur18} rely on this style of reasoning using similar independence relations.

Consider the program run illustrated in \figref{com}(a), where there is a race between the two red events.
\begin{wrapfigure}[15]{r}{0.4\textwidth}
\vspace{-10pt}
\includegraphics[scale=0.9]{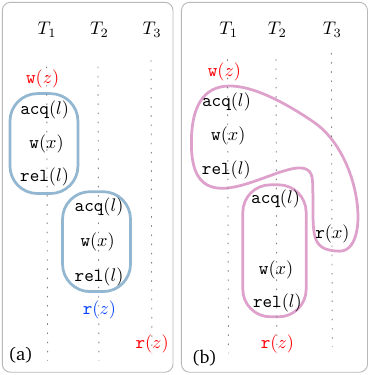}
\vspace{-18pt}
\caption{\small Commutativity-based races} 
\label{fig:com}
\end{wrapfigure}
This race can be predicted through the following commutativity argument: 
the $\textcolor{red}{\rd(z)}$ event of thread $T_3$ 
(denoted $\ev{T_3, \textcolor{red}{\rd(z)}}$)
can be commuted upwards against all actions of threads
$T_1$ and $T_2$ until it reaches $\ev{T_1, \textcolor{red}{\wt(z)}}$. 
This is based on the simple observation that any pair of events from different threads commute unless they share a memory location and at least one is a write. In this view, acquire and releases of locks can be treated as write accesses to the memory location corresponding to the lock. It is also known in the literature as a {\em happens-before} race. 

Now consider the event $\ev{T_2, \textcolor{blue}{\rd(z)}}$.
There is a predictable race between this event and 
the event $\ev{T_1, \textcolor{red}{\wt(z)}}$ as well, 
but a simple commutativity argument as described above cannot predict it. 
The accesses to lock $l$ and memory location $x$ are in the way and do not commute. 
In \cite{FarzanMathurPOPL2024}, a more general notion of commutativity is introduced, which can be used to predict this race. 
Its premise is that if the circled collections of events are considered as atomic {\em grains}, then, under the assumption that the two $\wt(x)$ events (in $T_1$ and $T_2$ respectively) are not observed by any reads not in the figure, {\em the two grains commute}, 
the events $\ev{T_1, \textcolor{red}{\wt(z)}}$ and $\ev{T_2, \textcolor{blue}{\rd(z)}}$
commute against these grains as before, and can be made concurrent.

Finally, consider the two events $\ev{T_1, \textcolor{red}{\wt(z)}}$ and $\ev{T_2, \textcolor{red}{\rd(z)}}$ in \figref{com}(b), which also form a predictable race. 
But, neither of the two commutativity-based techniques above can predict this race. 
In \cite{FarzanMathurPOPL2024}, the notion of {\em scattered grains} is introduced where you can pick a non-contiguous sequence of events as a (scattered) grain, 
as is the case for both pink blocks marked in the figure. 
Then, one can argue that in the absence of any latter $\rd(x)$ events that would observe either $\wt(x)$ event in the figure, the two (scattered) grains commute, and 
the two events $\ev{T_1, \textcolor{red}{\wt(z)}}$ and $\ev{T_2, \textcolor{red}{\rd(z)}}$ commute with them, witnessing the race. In Section \ref{sec:recap-commutativity}, we define the class of races that can be predicted using each commutativity principle (events, grains, and scattered grains), and argue that the class of races discovered using scattered grain commutativity is strictly larger than those discovered using grain commutativity, which is itself strictly larger than the class of races discovered using event commutativity.

\begin{wrapfigure}[11]{r}{0.16\textwidth}
\vspace{-10pt}
\includegraphics[scale=0.9]{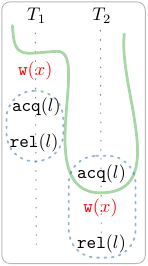}
\vspace{-20pt}
\caption{\scriptsize Prefix Races} 
\label{fig:syncp}
\end{wrapfigure}
The mechanism predicting races through {\em prefix reasoning} \cite{Mathur21,Tunc2023deadlock,Ang2024CAV} is fundamentally different. 
Consider the program run illustrated in \figref{syncp}, where there is a race between 
the events $\ev{T_1, \textcolor{red}{\wt(x)}}$ and $\ev{T_2, \textcolor{red}{\wt(x)}}$. First, observe that this race cannot be witnessed by any of the three commutativity-style techniques outlined above. For the race to be witnessed, the two critical sections on lock $l$ must be reordered. 
{\em Grain commutativity} (using the dashed grains) permits this reordering. However, after this reordering,  the grain containing the second $\ev{T_2, \textcolor{red}{\wt(x)}}$, as a whole, becomes adjacent to the event $\ev{T_1, \textcolor{red}{\wt(x)}}$; but the pair of red events cannot become adjacent. 

Prefix reasoning views the solid (green) curve as a {\em cut-off} point. 
If we consider the {\em prefix} before this point, which includes only a single $\acq(l)$ event from thread $T_2$, then this {\em prefix} is an executable run of the underlying program whenever the entire illustrated run is. 
After this prefix, both $\textcolor{red}{\wt(x)}$ events are enabled and can be executed simultaneously to cause a data race. 
Hence, the prefix marked by the green curve {\em witnesses the race}. 
In full generality, prefix reasoning seeks a cut-off point where the set of events before can be correctly and efficiently {\em linearized} into an executable run of the program.

In Section \ref{sec:recap-prefixes}, we first survey two existing such classes of prefixes in the literature. We then propose a new class which helps frame and explain the core principle behind this style of reasoning for predictive race detection. 
Then, with the aid of this new proposed class, we give a formal argument that the class of races predicted through prefix reasoning is strictly larger than the class of races predicted through scattered grain commutativity, and therefore all known commutativity-style reasonings for predicting races. This, for the first time, settles the question of which style has the better (theoretical) predictive power. 

The natural question, that comes next, motivates the main contribution of this paper: ``Can prefix and commutativity reasoning be combined to yield a new technique that overcomes the limitations of both techniques?''.
Moreover, ``can we arrive at a systematic and {\em modular} combination?''. To this end, one hopes to find a way for the techniques to complement one another and partially compensate for each other's limitations.

The limitations to commutativity reasoning are easier to formulate: a race cannot be predicted between a pair of events if there is at least one event in the middle of the pair that cannot be commuted out of the range. For understanding the limitations of prefix reasoning, consider the runs illustrated in \figref{syncp-miss}, which are 
slight modifications of the one from \figref{syncp}. In both (a) and (b), 
\begin{wrapfigure}[14]{r}{0.43\textwidth}
\vspace{-10pt}
\includegraphics[scale=0.9]{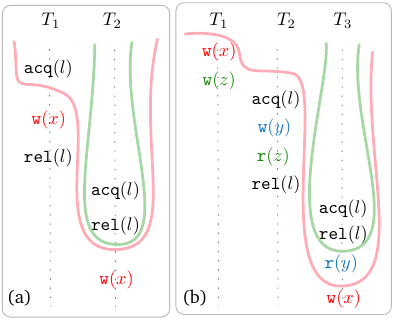}
\vspace{-15pt}
\caption{Prefix reasoning limitations} 
\label{fig:syncp-miss}
\end{wrapfigure}
the pair of events on location $\textcolor{red}{x}$ form a predictable race, yet neither race can be predicted by prefix reasoning. In both cases, the outer (pink) curves mark the cut-off point after which the pair of racy events are at the boundary of the prefix. However, each case is problematic in a different way. In \figref{syncp-miss}(a), the outer (pink) curve marks a prefix that is not executable, due to an inconsistency with lock semantics: while the first lock block remains open, the second lock block cannot be executed. In \figref{syncp-miss}(b), the outer (pink) curve marks an executable prefix, but the event $\ev{T_3, \textcolor{red}{\wt(x)}}$ is not {\em enabled} after this prefix. Note that the prefix excludes the event $\ev{T_2, \textcolor{blue!70!green}{\wt(y)}}$ from which the read event $\ev{T_3, \textcolor{blue!70!green}{\rd(y)}}$ observes its value in the given run. In the execution of this prefix, the value of this read may be different, which implies that one cannot soundly rely on the thread to follow the same local execution path as before. If one attempts to correct this by including $\ev{T_2, \textcolor{blue!70!green}{\wt(y)}}$ in the prefix, then one has to include the entire lock block from $T_2$ (to respect lock semantics for executability), which in turn triggers the inclusion of the two events of $T_1$ to maintain the executability of the prefix, and it will no longer witness the race.

It is easy to speculate that limitations of prefix reasoning can be overcome through richer classes of prefixes. For instance, 
looking at \figref{syncp-miss}(a), the reader may wonder that if we reorder the elements in the prefix, specifically by executing the lock block in thread $T_2$ first, then we arrive at an executable prefix that witnesses this race, and whether this additional reasoning can be done through an efficient algorithm. In \cite{OSR2024}, a partial solution is provided in the form of a heuristic algorithm that can do some partial reasoning of this kind. However, the set of predictable races are incomparable against any of aforementioned (prefix and commutativity reasoning) techniques. In particular, this heuristic does not even subsume vanilla prefix reasoning. 
More importantly, commutativity and prefix reasoning techniques for race prediction all have constant-space (wrt the length of the input run) complexity, while reasoning about more elaborate prefixes makes
reasoning harder~\cite{Mathur21,Ang2024CAV}; 
indeed the algorithm of \cite{OSR2024} has linear space complexity. 

Nevertheless, the existence of this heuristic poses the more specific question: If one uses commutativity reasoning in the prefixes, to expand the class of prefixes, would one arrive at a strictly better predictive algorithm? Somewhat surprising, this is not the case.  
In Section \ref{sec:combining}, we argue that using (event, grain, or scattered grain) commutativity in defining a richer class of prefixes does not yield any additional power to a race prediction algorithm that uses these prefixes for the means of race prediction.  Intuitively, the prefix up to the cut-off point is executable if and only if any up-to-commutativity reordering of it is, and the set of events enabled at the boundary do not change.  Therefore, using commutativity \emph{inside} the prefix does not yield any new executable prefixes or any new races at the end of existing ones.

What if we use commutativity reasoning {\em after} the prefix? Consider the prefix that is marked by the inner (green) curve in \figref{syncp-miss}(a). This prefix is executable, but does not witness any race; the two events enabled immediately after it are $\ev{T_1, \acq(l)}$  and $\ev{T_2, \textcolor{red}{\wt(x)}}$, and do not constitute a race. However, it is straightforward to reason, using event commutativity why there is a race in the remaining executable suffix: the event $\ev{T_2, \textcolor{red}{\wt(x)}}$ commutes against the event $\ev{T_1, \rel(l)}$ 
and can be brought next to the event $\ev{T_1, \textcolor{red}{\wt(x)}}$, {\em in the suffix}. But, for this reasoning to kick in, one needs to first get rid of the prefix marked by the inner curve. The same scenario plays out if we use the prefix marked by the inner (green) curve in \figref{syncp-miss}(b). 

In Section \ref{sec:suffix}, we present our first approach for combining prefix reasoning and commutativity reasoning in tandem, yielding a constant-space algorithm for the results in Section \ref{sec:algo}. Prefix reasoning contributes by removing some obstacles that commutativity reasoning cannot overcome alone, and event-based commutativity reasoning {\em in the remaining executable suffix} (i.e., \emph{outside} the prefix) adds a new dimension of expressiveness to races that would otherwise be missed by prefix reasoning alone, as illustrated by the examples in \figref{syncp-miss}. However, as we argue in Section \ref{sec:suffix}, this new class of races does not strictly subsume all prefix races. In a sense, this algorithm suggests a 
\begin{wrapfigure}[15]{r}{0.22\textwidth}
\vspace{-10pt}
\includegraphics[scale=1]{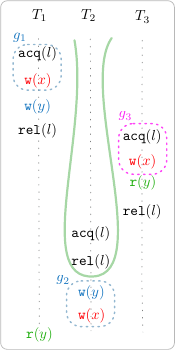}
\vspace{-10pt}
\caption{\footnotesize Grains + Prefixes} 
\label{fig:grains}
\end{wrapfigure}
new point of expressiveness in the style of \cite{OSR2024}: it can beat vanilla prefix reasoning in some instances (e.g. the races in \figref{syncp-miss}(a,b)), but it can also be beaten by vanilla prefix reasoning. The distinction is that unlike \cite{OSR2024}, it admits a constant-space prediction algorithm. 

Consider the example run illustrated on the right. The prefix marked with the green curve is executable, however, the race between the events $\ev{T_1, \textcolor{red}{\wt(x)}}$ and $\ev{T_2, \textcolor{red}{\wt(x)}}$ in the executable suffix (consisting of all the remaining events in this example) cannot be witnessed using event-based commutativity;
the pair of $\textcolor{blue!70!green}{\wt(y)}$ events become a commutativity obstacle to this race. 
If we consider the suffix without the very last $\textcolor{green!70!blue}{\rd(y)}$ event and all events from $T_3$, then grain commutativity can predict this race, because as singleton grains (without any future reads), the two
$\textcolor{blue!70!green}{\wt(y)}$ events commute. 
But, even only with the addition of
the last $\textcolor{green!70!blue}{\rd(y)}$ event, no existing commutativity reasoning technique can predict the race in the suffix. 
Also, consider the two grains $g_1, g_2$ marked (with dashes) in the figure, and imagine that in the spirit of grain commutativity, we consider these two grains as two compound events. 
These two compound events are {\em enabled} after the green prefix. 
So, if we were to shift our view from individual events to grains, we could declare the two grains to be racy witnessed by this prefix. 
Once we have this fact, we can zoom into the segment of the suffix (instead of the complete suffix) that consists only of the concatenation of these two grains $g_1, g_2$: 
$$\ev{T_1, \acq(l)} \ev{T_1,\textcolor{red}{\wt(x)}} \ev{T_2,\textcolor{blue!70!green}{\wt(y)}} \ev{T_2,\textcolor{red}{\wt(x)}}$$  
in which the race can be predicted with a single event commutation. There is a similar situation with the race between $\ev{T_3, \textcolor{red}{\wt(x)}}$ and $\ev{T_2, \textcolor{red}{\wt(x)}}$. But, a different grain, namely $g_3$, combined with $g_2$ witnesses this race. The choices of $g_1$ and $g_3$ do not witness any race, simply because the pair of events $\ev{T_1, \textcolor{red}{\wt(x)}}$ and $\ev{T_3, \textcolor{red}{\wt(x)}}$ do not form a race.

Inspired by this observation, we introduce the main contribution of 
this paper --- {\em granular prefix races} --- which is a class of races that can be efficiently (in constant space) predicted using this granular view of prefix reasoning, combined with fast commutativity reasoning (see Section \ref{sec:granular}) for the {\em last-mile} reasoning within the two grains. Granular prefix races contain all prefix races, and as the example illustrates, this containment is strict. Note that beyond the default enumeration of possible prefix choices, granular prefix reasoning enumerates the choices of grains, since different choices may witness different races missed by vanilla prefix reasoning, as the example demonstrates.


In \cite{Ang2024CAV}, it was observed that a constant space algorithm for prefix reasoning alone can behave poorly in practice. Intuitively, think of this algorithm as guessing all possible prefixes, which are maintained as a constant-bounded set of summaries. This constant state space has to be carefully maintained every time a new event is processed by the algorithm, and the price of this maintenance, even for modest-sized space, over millions of events does not yield a practically efficient algorithm. Inspired by antichain techniques \cite{anitChain2006} from automata theory, we propose a new version of this algorithm which substantially cuts down on this price (see Section \ref{sec:algo}). This improvement is vital, since our granular prefix reasoning builds on this baseline algorithm. We also adapt the idea of antichain techniques \cite{anitChain2006} to extend the optimization to the new algorithm for predicting granular prefix races (Section \ref{sec:algo}).

One key advantage of granular prefix reasoning, for combining the two styles of reasoning---prefix/suffix and commutativity--- based on grains, is that it yields opportunities in devising principled compromises in expressiveness to regain algorithmic scalability: one can tune the algorithm effort based on the kinds of grains that are enumerated to strike a balance between expressiveness and efficiency (see Section \ref{sec:exp-result}).

To summarize, our key contributions are:
\begin{itemize}
	\item We define three classes of data races based on \emph{commutativity} reasoning of
	increasing granularity --- event, grain and scattered grains.
	We then introduce a principled approach to formulate
	different classes of races based on \emph{prefix} reasoning, present them in a unified setting,
	and also present a new and simpler class of races that coincides with existing notions.
	We then compare the predictive power of all these classes of data races 
	and outline the key message that prefix reasoning is more powerful than commutativity reasoning, when used in isolation \secref{commutativity-prefixes}.

	\item We then study when and how can these two reasoning techniques be modularly combined.
	We show that combining commutativity inside the prefix does not enhance predictive power (\secref{combining}).
	We then show that commutativity can enhance predictive power if used beyond the prefix,
	and propose two new classes of races, {\em maximal suffix} and {\em granular prefix} based on this principle (\secref{stratification}).

	\item We devise efficient (constant-space linear-time) algorithms for the 
	prediction of  {\em granular prefix} races, and present an antichain optimization to aid practical performance (\secref{algo}). 

	\item We implement our algorithms in Java and put them to test with a thorough evaluation of them on benchmark suites derived from prior works on data race prediction, demonstrating the effectiveness
	of our new notion of $\grconfp$-races, the proposed algorithms, optimizations and heuristics (\secref{experiments}).
\end{itemize} 


\section{Preliminaries}
\seclabel{prelim}

In this section, we discuss notations on shared-memory 
multithreaded concurrent programs and formally define data races
and predictive data races.

\subsection{Concurrent Programs and Data race prediction}

\myparagraph{Runs and events}{
	In this work, we consider shared-memory multithreaded concurrent programs 
	that work under sequential consistency.
	An execution or \emph{run} $\tr$ of a concurrent program is a sequence of \emph{events} $e_1 e_2 \dots e_n$ performed by a finite set of threads $\threads$.
	Each event either accesses (reads from or writes to) one of
	the shared memory locations $\mems$ or acquires or releases one of the
	locks $\locks$ for enforcing mutual exclusion;
	for ease of presentation we skip other kinds of synchronizations such as barriers
	which can easily be modeled in our setting.
	Formally, an event is then a tuple $e = \tuple{id, lab}$, 
	where $id$ is a unique identifier for $e$
	and $lab \in \labs$, where $\labs = \labs_\mem \uplus \labs_\lck$,
	and
	\begin{align*}
	\begin{array}{rcl}
	\labs_\mem &&= \setpred{\ev{t, op(x)}}{t \in \threads, op \in \set{\wt, \rd}, x \in \mems}\\ 
	\labs_\lck &&= \setpred{\ev{t, op(\lk)}}{t \in \threads, op \in \set{\acq, \rel}, \lk \in \locks}
	\end{array}
	\end{align*}
	We will refer to the thread identifier, operation and memory location (or lock) accessed
	in an event $e$ labeled with $lab = \ev{t, op(d)}$ by $\ThreadOf{e} = t$, $\OpOf{e} = op$ and 
	$\ObjOf{e} = d$;
	when $op \in \set{\wt, \rd}$, then we sometimes
	use the notation $\MemOf{e}$ instead of $\ObjOf{e}$,
	and when $op \in \set{\acq, \rel}$, then we
	use the notation $\LockOf{e}$ instead of $\ObjOf{e}$.
	Often, the unique identifier $id$ of an event $e$ will be clear from context or entirely irrelevant.
	We will not mention it explicitly, and will instead write $e = \ev{t, op(d)}$, 
	where $\ev{t, op(d)}$ is the label of $e$.

	We use $\events{\tr} = \set{e_1, \ldots, e_n}$ to denote the set of 
	events of the run $\tr = e_1 e_2 \dots e_n$, 
	and $\trord{\tr} = \setpred{(e_i, e_j)}{e_1, e_j \in \events{\tr}\text{ and }i < j}$ to denote the total order on $\events{\tr}$ induced by the sequence $\tr$.
	The program order $\po{\tr}$ of a run $\tr$ is the smallest partial order 
	that includes pairs $(e, f)$ in $\tr$ when $\ThreadOf{e} = \ThreadOf{f}$ and $e \trord{\tr} f$.
	The \emph{reads-from} relation $\rf{\tr}$ of $\tr$ is the set of all memory access pairs
	$(e_w, e_r)$ in $\tr$ such that 
	$\OpOf{e_w} = \wt$, $\OpOf{e_r} = \rd$, $\MemOf{e_w} = \MemOf{e_r}$, $e_w \trord{\tr} e_r$, 
	and for every other write $e'_w \neq e_w$ ($\OpOf{e'_w} = \wt$) with $\MemOf{e'_w} = \MemOf{e_w}$,
	we have either $e'_w \trord{\tr} e_w$ or $e_r \trord{\tr} e'_w$.
	We will often use $\proj{\tr}{t}$, $\proj{\tr}{\lk}$ and $\proj{\tr}{x}$ to denote
	the projection of $\tr$ to the set of events respectively performed by some thread $t \in \threads$,
	accessing a lock $\lk \in \locks$ and accessing a memory location $x \in \mems$.
	A concurrent program run $\tr$ is said to be \emph{well-formed} when,
	\begin{enumerate*}[label=(\alph*)]
		\item each read event has a \emph{corresponding} write event, i.e., 
		for each $x \in \mems$, $\proj{\tr}{x}$ is of the form $(\wt(x) \cdot (\rd(x))^*)^*$, and
		\item critical sections on the same lock do not overlap, i.e.,
		for each $\lk \in \locks$, $t\in\threads$, $\proj{\tr}{\lk, t}$ is of the form 
		$(\acq(\lk) \cdot \rel(\lk))^*(\acq(\lk) + \varepsilon)$.
	\end{enumerate*}
	We will assume runs are well-formed from now on.

}

\myparagraph{Data races}{
	Data races are one of the most common concurrency bugs and are indicative of possibly more 
	serious issues such as memory corruption and security vulnerabilities, and proactive
	detection of data races has been proven effective in isolating bugs
	early on during the development cycle.
	Here we focus on dynamic analysis algorithms that analyze program executions and check 
	if they contain data races.
	While many notions of data races have been proposed in the literature, 
	here we present the most popular one used in prior works on data race detection~\cite{Mathur18,Kini17,Smaragdakis12}.
	At a high level, a data race occurs in an execution if
	two conflicting events occur \emph{simultaneously} in it.
	A pair of events $(e_1, e_2)$ in an execution $\tr$
	is said to be conflicting if they access the same memory location and 
	at least one of them is a write operation (formally, $\MemOf{e_1} = \MemOf{e_2} = x$ and $\set{\wt} \subseteq \set{\OpOf{e_1}, \OpOf{e_2}} \subseteq \set{\rd, \wt}$).
	In the setup we have, simultaneity can be modeled by 
	instead asking if two such events are consecutive.
	We thus have the following.
	In a concurrent program run $\tr$, 
	a pair of conflicting events $(e_1, e_2)$ in $\tr$ 
	is said to be a data race if $\ThreadOf{e_1} \neq \ThreadOf{e_2}$,
	and they appear consecutively in $\tr$.
	A run $\tr$ is said to have a data race if it contains one.
}

\myparagraph{Correct reorderings, enabled events and predictable data races}{
	While the above definition of data races immediately lends itself to a simple
	algorithm for automatically detecting data races from program runs, such an algorithm
	is likely to miss many races due to its reliance on an angelic thread interleaving that
	puts conflicting events next to each other.
	In contrast, \emph{predictive} style of reasoning takes a slightly different 
	approach~\cite{Said11,serbanuta2013maximal}, examining not only the observed run but also inferring alternative feasible executions.
	A well studied notion of the space of alternative executions is that of
	\emph{correct reorderings}~\cite{Smaragdakis12,Mathur18} of the observed run $\tr$.
	Formally, the set $\creorderings{\tr}$ of correct reorderings of a well-formed run $\tr$
	can be defined to be the set of all well-formed runs $\rho$ such that
	\begin{enumerate*}
	\item $\events{\rho} \subseteq \events{\tr}$,
	\item for each thread $t \in \threads$, $\proj{\rho}{t}$ is a prefix of $\proj{\tr}{t}$,
	\item $\rf{\rho} \subseteq \rf{\tr}$.
	\end{enumerate*}

	Armed with this definition, one can define a more general but still 
	robust definition of \emph{predictable} data races as follows.
	A pair of conflicting events $(e_1, e_2)$ in $\tr$ 
	is said to be predictable data race if they are $\tr$-enabled in a correct reordering $\rho$.
	Here, we say that an event $e\in\events{\tr}$ is \emph{$\tr$-enabled} in a correct reordering
	$\rho$ if $e \not\in \events{\rho}$ and for all events $e' \in \events{\tr}$ 
	such that $(e', e) \in \po{\tr}$, we have $e' \in \events{\rho}$.
	Notably, correct reorderings preserve both program order and data/control flow of $\tr$. 
	This preservation ensures that any program $P$ generating $\tr$ must also be capable of generating all its correct reorderings. 
	This property forms the foundation for sound data race prediction: algorithms that analyze $\tr$ and search for race witnesses in $\creorderings{\tr}$ are guaranteed to report only true positives. Since nearly all races discussed in this paper are predictable races, to avoid tedium, we simply refer to a {\em predictable race} as a race.



}




\section{The role of commutativity and prefixes in predictive analysis}
\seclabel{commutativity-prefixes}
In this section, we identify two distinct principles that yield linear time and
constant space algorithms for predictive data race detection: \emph{commutativity} and \emph{prefix reasoning}.
We demonstrate these two principles next,
in the context of data race prediction and compare their expressive power. In the process, we expose the principles behind an array of data race prediction techniques that may otherwise look ad hoc. 


\subsection{Commutativity-based Reasoning}
\seclabel{recap-commutativity}

The key principle behind commutativity reasoning is simple --- infer an equivalent
correct reordering via repeated commutations of atomic elements
of an execution.
%
Mazurkiewicz's trace theory~\cite{Mazurkiewicz87} provides a classical
framework for commutativity reasoning when atomic elements are 
chosen to be individual events in the execution.
We briefly recall this next, and subsequently recall recent
generalizations to the case where the choice of atomic elements 
includes larger subsets of events, called \emph{grains}, 
allowing for the possibility of
improved predictive power~\cite{FarzanMathurPOPL2024}.

\myparagraph{Event-based commutativity}{
	To formally describe trace equivalence, one
	first fixes a symmetric, irreflexive \emph{independence} relation
	$\indrel \subseteq \labs \times \labs$ on the set of event labels $\labs$.
	With this, executions $\tr$ and $\rho$ are said to be trace-equivalent,
	denoted $\tr \mazeq \rho$, if
	$\tr$ can be transformed into
	$\rho$ by repeatedly swapping consecutive events
	labeled $a, b \in \labs$ so that $(a, b) \in \indrel$\footnote{More formally, $\mazeq$ is the smallest equivalence on $\labs^*$ such that
	for any two words $w_1, w_2 \in \labs^*$ and for each $(a, b) \in \indrel$, we have:
	$w_1 \concat a \concat b \concat w_2 \mazeq w_1 \concat b \concat a \concat w_2.$ We omit explicit parametrization on the independence relation $\indrel$ from our notation $\mazeq$ (i.e., avoid cumbersome notations like $\equiv_{\maz, \indrel}$ or $\equiv_{\maz}^\indrel$) since it will often be clear from context.}.
	We use $\mazcl{\tr}$ to denote the set of all executions equivalent to $\tr$ by $\mazeq$.
	Trace equivalence is the simplest form of commutativity reasoning
	and can help establish pair of events to be in race if they can be brought together
	by repeated commutations of neighboring independent events; we call such races $\maz$-races.
	A pair $(e_1, e_2)$ of conflicting events is a Mazurkiewicz-race, or $\maz$-race in $\tr$ if
	there is a $\rho \mazeq \tr$ such that $e_1$ and $e_2$ appear
	consecutively in $\rho$.
}

\myparagraph{Soundness of $\maz$-races}
For the above scheme --- push events either before $e_1$ or after $e_2$
through repeated commutations --- to be sound and effective, 
one must choose the independence relation $\indrel$ carefully.
In particular, an overly permissive $\indrel$ may result into a reordering
that is not a correct reordering (i.e., it may not be sound), while an overly conservative $\indrel$
may forbid most commutations and would not be useful.
For the alphabet $\labs_\mem \uplus \labs_\lck$,
we say that $\indrel$ is said to be \emph{sound} if for every well-formed execution
$\tr$, we have $\mazcl{\tr} \subseteq \creorderings{\tr}$.
Naturally, an $\maz$-race is a (predictable) race if $\indrel$ is sound.
The most permissive sound choice of $\indrel$ for the alphabet 
$\labs_\mem \uplus \labs_\lck$
is given by:
\begin{align*}
\indrel = \setpred{(a_1, a_2)}{
\ThreadOf{a_1} \neq \ThreadOf{a_2} \land 
\big(\ObjOf{a_1} = \ObjOf{a_2} \implies \OpOf{a_1} = \OpOf{a_2} = \rd \big)
}
\end{align*}
%
Unless otherwise stated, we will assume that the independence relation is as above.
As an example, recall the execution in \figref{com}(a), and events $\ev{T_1, \wt(z)}$ and
$\ev{T_3, \rd(z)}$.
Here, since the latter is independent of all events in this execution, except $\ev{T_1, \wt(z)}$, it can be swapped against them to predict the race, which is a $\maz$-race.


\myparagraph{Grain Commutativity}
Reasoning based solely on event-based commutativity, a la trace
equivalence, is known to be very conservative 
and misses out on many data races in practice~\cite{Said11,Smaragdakis12}.
The fundamental limitation of sound event-based commutativity
arises from the fact that it  only allows those commutations that are sound at each step.
As we noted with the example run in \figref{com}(a), the race between the red $\wt(z)$ and the blue $\rd(z)$ can be uncovered
by commuting the critical sections, as a whole \emph{grains}
against each other.
\emph{Grain equivalence}~\cite{FarzanMathurPOPL2024} essentially
formalizes this notion as a natural generalization of trace equivalence.
In essence, an execution $\rho$ can be obtained from $\tr$ using
grain commutativity, denoted $\rho \in \graincl{\tr}$,
if there is a partition of $\tr$ into grains, or contiguous
sequences of events ($\tr = g_1 g_2 \cdots g_k$)
such that $\rho$ can be obtained by repeated commutations of these grains
according to a \emph{grain independence relation} $\indrel_\grains$.
As before, the largest sound grain independence relation is unique for a choice of grains;
we skip the detailed definition here and assume $\indrel_\grains$ is this largest independence relation.
With this, can now define a race that can be inferred using grain commutativity reasoning --- a pair
of conflicting events $(e_1, e_2)$ is a $\grains$-race in $\tr$
if there is a $\rho \in \graincl{\tr}$ such that $e_1$ and $e_2$ are 
consecutive in $\rho$.
Thus, in the run in \figref{com}(a),
the events $\ev{T_1, \wt(z)}$ and $\ev{T_2, \rd(z)}$ constitute a $\grains$-race.

\myparagraph{Scattered-grain commutativity}
Finally, \emph{scattered grains}~\cite{FarzanMathurPOPL2024} allow for 
commuting subsequences of events which may not be contiguous.
The formal definition of a data race that can be inferred using scattered
grain commutativity can be given in terms of a grain graph induced by a 
given choice of scattered grains.
Let $S = \set{g_1, g_2, \ldots, g_k}$ 
be a set of pairwise disjoint subsequences of events, or \emph{scattered grains} 
in $\tr$ such that $\events{\tr} = \biguplus_{i=1}^k \events{g_i}$.
The grain graph $\mathsf{GGraph}_{\tr, S} = (S, E)$ adds an edge from a grain $g_i$
to a later grain $g_j$ if there is a dependence between them.
The grain graph captures \emph{causal concurrency} ---
it is sound to conclude that grain $g$ can be reordered before $g'$ if there is no path from
$g'$ to $g$ in the graph $\mathsf{GGraph}_{\tr, S}$.
Indeed, let $\set{\scc_1, \scc_2, \ldots, \scc_m}$ be the
strongly connected components of $\mathsf{GGraph}_{\tr, S}$.
Then, any topological ordering $\scc_{i_1} \cdot \scc_{i_1} \cdots \scc_{i_m}$ 
of the condensation (obtained after contracting the SCCs into single vertices) 
of this graph can be used to obtain
a sound reordering $\rho$, given by the concatenation
$\rho = \textsf{lin}(\scc_{i_1}) \textsf{lin}(\scc_{i_1}) \cdots \textsf{lin}(\scc_{i_m})$,
where $\textsf{lin}(\scc_{i_j})$ is the sequence obtained by arranging
the events in $\bigcup_{g \in \scc_{i_j}} \events{g}$ according to their order in $\tr$.
We let $\scgraincl{\tr}^S$ to denote all reorderings obtained from $\tr$
in this manner, using $S$ as the choice of scattered grains.
With this, we can now define a data race as follows.
A pair of conflicting events $(e_1, e_2)$ is said to be a scattered-grain race, or
a $\scatteredgrains$-race of $\tr$, if there is a choice of grains $S$
and an execution $\rho \in \scgraincl{\tr}^S$
such that $e_1$ and $e_2$ are consecutive in $\rho$.

The predictive power and soundness of data race detection based on the above notions of
commutativity can be summarized as follows.

\begin{restatable}{proposition}{commutativityComparePredictivePower}[Predictive Power of Commutativity Reasoning]
\proplabel{commutativity-soundness-predictive-power}
For any given program run $\tr$,  
the set of $\maz$-races of $\tr$ is strictly contained in the set 
$\grains$-races of $\tr$, which is itself strictly contained in the set of $\scatteredgrains$-races of $\tr$, each of which is a predictable race. 
\end{restatable}

In~\cite{FarzanMathurPOPL2024}, it is argued how commutativity reasoning can yield 
efficient algorithms for determining \emph{causal concurrency} between events~\cite{FarzanMathurPOPL2024}.
These algorithms can be modified to also obtain efficient algorithms
for data race prediction, giving us the following holy grail result of monitorability;
here, we assume $|\labs|$ is constant.

\begin{restatable}{theorem}{commutativityRaceConstantSpace}
\thmlabel{commutativity-race-constant-space}
Let $C \in \set{\maz, \grains, \scatteredgrains}$ be one of the commutativity 
granularities discussed above.
The problem of checking if an execution $\tr$ has a $C$-race can be solved
using a streaming algorithm that takes constant space and $O(|\tr|)$ time.
\end{restatable}



\subsection{Prefix Reasoning}
\seclabel{recap-prefixes}

Reasoning based on prefixes can generally be used for any specification
that asks if a set of events are simultaneously enabled in some correct reordering. Generic specifications like this are useful for predicting races, but also other things like deadlock detection~\cite{Tunc2023deadlock}.

Prefix-based reasoning looks for an appropriate
subset $S \subseteq \events{\tr}$ of events of the execution $\tr$ 
such that $S$ is downward closed w.r.t. $\po{\tr}$ (hence `prefix'),
the two given conflicting events $e_1$ and $e_2$
are enabled in $S$, and further, there is a linearization of $S$ that is a correct reordering of $\tr$.
A careful reader may observe that, as such, this broad description
of prefix reasoning in fact includes the entire class of predictive races,
and thus, in its full generality, looking for such a set $S$ and its linearization
is intractable~\cite{Mathur2020b}.
In response, recent works have identified specific classes of 
\emph{linearizations} for the set $S$,
that help retain tractability~\cite{Mathur21,Tunc2023deadlock,Ang2024CAV}.
Here, we present the otherwise disparate notions in a uniform, systematic manner
as instances of prefix reasoning, and also
introduce a new class of races (\defref{seqp-race}) based as another instance of this uniform presentation.

\myparagraph{Synchronization-preserving prefixes and data races}
	Synchronization-preserving (or $\syncp$ for short) data races, 
	recently identified in~\cite{Mathur21} are those that are enabled
	at the end of a $\syncp$-prefix.
	Formally, a $\syncp$-prefix $\rho$ of an execution $\tr$
	is a correct reordering $\rho$ of $\tr$ such that
	for any two acquire events $a_1, a_2 \in \events{\tr}$ 
	on the same lock
	(i.e., $\OpOf{a_1} = \OpOf{a_2} = \acq$, 
	$\LockOf{a_1} = \LockOf{a_2} = \lk$), 
	whenever $a_1, a_2 \in \events{\rho}$, 
	then, $a_1 \trord{\rho} a_2$ iff $a_1 \trord{\tr} a_2$. 
	In other words, a $\syncp$-prefix preserved the order
	of same-lock acquire events that are retained in the reordering,
	but may flip the relative order between other events (including conflicting pairs of events).
	A pair of conflicting events $(e_1, e_2)$ in $\tr$ 
	is a $\syncp$-race of $\tr$ 
	if there is a $\syncp$-prefix
	$\rho$ of $\tr$ in which $e_1$ and $e_2$ are both $\tr$-enabled.
	Recall the example execution illustrated in \figref{syncp}. The green curve marks a $\syncp$-prefix, after which the
	two events $\ev{T_1, \wt(x)}$ and $\ev{T_2, \wt(x)}$ are enabled. 
	$\syncp$-races can be detected using a linear time and linear space algorithm~\cite{Mathur21}.

\myparagraph{Conflict-preserving data races}
	An execution $\rho$ is said to be a conflict-preserving prefix, 
	or $\confp$-prefix,
	of execution $\tr$ if $\rho$ is a correct reordering of $\tr$
	and further $\rho \mazeq \proj{\tr}{\events{\rho}}$, where
	$\proj{\tr}{E}$ is the projection of $\tr$ to the set $E$.
	That is, the relative order between events of $\rho$ and the events
	of $\tr' = \proj{\tr}{\events{\rho}}$ is the same if these events are dependent;
	otherwise their relative order may change.
	A $\confp$-race of $\tr$ is then a pair of conflicting events 
	$(e_1, e_2)$ in $\tr$ 
	such that there is a $\confp$-prefix
	$\rho$ of $\tr$ in which $e_1$ and $e_2$ are both $\tr$-enabled~\cite{Ang2024CAV}.
	Observe that every $\confp$-prefix is also a $\syncp$-prefix and thus
	every $\confp$-race is also a $\syncp$-race by definition.
	More importantly though, for the case of data race prediction, the smaller
	class of $\confp$-prefixes is sufficient to, in fact, detect all $\syncp$-races.
	That is, every $\syncp$-race is also a $\confp$-race~\cite{Ang2024CAV}.
	Indeed, in the run of \figref{syncp},
	the single event prefix marked with the (green) curve is also a $\confp$-prefix
	and thus the events $\ev{T_1, \wt(x)}$ and $\ev{T_2, \wt(x)}$ also constitute a $\confp$-race.
	Finally, $\confp$-races can be detected in constant space and linear time~\cite{Ang2024CAV}.
	While theoretically more efficient than the linear space algorithm of $\syncp$-races,
	the proposed constant space automata-theoretic algorithm for detecting $\confp$-races
	relies on an on-the-fly membership check in an NFA with large state space,
	and can be slow in practice when the size of the alphabet $\labs$ 
	is moderately large~\cite{Ang2024CAV}.


\myparagraph{Race prediction using simpler prefixes}
	In principle, for a run $\tr$, the set of its $\syncp$-prefixes of
	$\tr$ is strictly larger than the set of its $\confp$-prefixes,
	and yet each race that can be detected using a $\syncp$-prefix can also be
	detected using a $\confp$-prefix. 
	In this work we show that such races can in fact be detected by an even smaller class of prefixes.
	In essence, this class of prefixes simply preserves the order of
	events as in the original execution and disallow all reorderings between
	events. We use \emph{sequential-order-preserving prefix} or $\seqp$-prefix
	to denote each such prefix (formally defined next).
	We denote the class of races witnessed using these prefixes 
	simply as \emph{prefix}-races\footnote{We choose the simpler nomenclature of \emph{prefix}-races instead of
	something like $\seqp$-races. As we show later in \propref{seqp-prefix-race-is-syncp-race}, all
	prior known classes of races ($\syncp$-races and $\confp$-races) 
	based on  prefix reasoning are subsumed by this class.
	In light of this, we decided to reduce the burden of additional cumbersome qualifiers 
	to the name of this class of races and opt for a simpler name that
	accurately represents the true expressive power of this class.
	}.
	\begin{definition}[Sequential-Order-Preserving prefix and prefix-races]
	\deflabel{seqp-race}
	An execution $\rho$ is a sequential-order-preserving prefix, $\seqp$-prefix, of execution $\tr$ if
	$\rho$ is a correct reordering of $\tr$ and for every $e, e' \in \events{\rho}$,
	we have $e \trord{\rho} e'$ iff $e \trord{\tr} e'$.
	A pair of conflicting events $(e_1, e_2)$ is said to be a prefix-race 
	of $\tr$ if there is a $\seqp$-prefix $\rho$ of $\tr$ such that both $e_1$
	and $e_2$ are $\tr$-enabled in $\rho$.
	\end{definition}
	Since a $\seqp$-prefix is also a $\confp$-prefix (which in turn is also a $\syncp$-prefix), 
	every prefix-race is also a $\confp$-race (and thus also a $\syncp$-race). 
	We show that even the converse is true:
	
	\begin{restatable}{proposition}{seqpRaceIsSyncpRace}
	\proplabel{seqp-prefix-race-is-syncp-race}
	Let $\tr$ be an execution and let $e_1$ and $e_2$ be conflicting events of $\tr$.
	$(e_1, e_2)$ is a prefix-race iff it is a $\confp$-race iff it is a $\syncp$-race.
	\end{restatable}


	As a result, a prefix-race can also be detected using a streaming
	constant space linear time algorithm, since $\confp$-races were shown
	to admit such an algorithm as well~\cite{Ang2024CAV}.

	\begin{restatable}{theorem}{seqpRaceConstantSpace}
	The problem of checking if an execution $\tr$ has a prefix-race can be solved
	using a streaming algorithm that takes constant space and $O(|\tr|)$ time.
	\end{restatable}


	Though both constant space algorithms, the algorithm for prefix-races is simpler than
	that for $\confp$-races since it does not have to guess a $\mazeq$-equivalent
	reordering of a selected set of events.
	At a high level, this algorithm essentially `guesses' a pair $(e_1, e_2)$
	of conflicting events, and also a prefix $\rho$
	by guesses the events of $\rho$,
	and checks if the guess is consistent --- the write event corresponding to
	each read event is in $\rho$ and for each lock $\lk$, only the last acquire event on $\lk$
	is allowed to be unmatched in $\rho$ --- and if the two events $e_1$ and $e_2$ 
	are enabled at the end of $\rho$.
	Since the guesses can be made in constant space, the result follows.


\myparagraph{Comparison with commutativity reasoning}
	$\syncp$-based (and thus also $\confp$ and $\seqp$-based)
	reasoning is known to be more permissive than reasoning based on
	event-based commutativity. 
	That is, every $\maz$-race is also a $\syncp$-race (alternatively, $\confp$-race or 
	prefix-race), but the converse is not true~\cite{Mathur21}.
	Does this change when we enhance the commutativity granularity
	from event-based to the more permissive notions of grain-based
	commutativity?
	Here, we show that prefix based reasoning strictly
	subsumes all the commutativity-based reasonings we discussed
	in \secref{recap-commutativity}:

	\begin{restatable}{theorem}{PrefixSubsumesCommutativity}
	\thmlabel{prefix-subsumes-commutativity}
	Let $C \in \set{\maz, \grains, \scatteredgrains}$ be a commutativity granularity.
	For every execution $\tr$, the set of $C$-races is strictly contained in the set of prefix-races.
	\end{restatable}

\begin{figure}[t]
	\centering
	\begin{subfigure}{0.28\textwidth}
		\centering
		\includegraphics[width=0.68\textwidth]{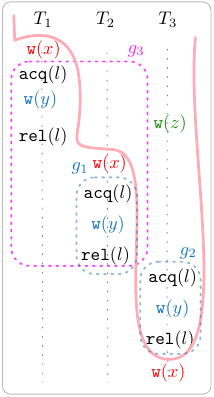}
		\caption{Commutativity reasoning}
		\figlabel{exp-commutativity}
	\end{subfigure}
	\begin{subfigure}{0.7\textwidth}
		\centering
		\includegraphics[width=0.8\textwidth]{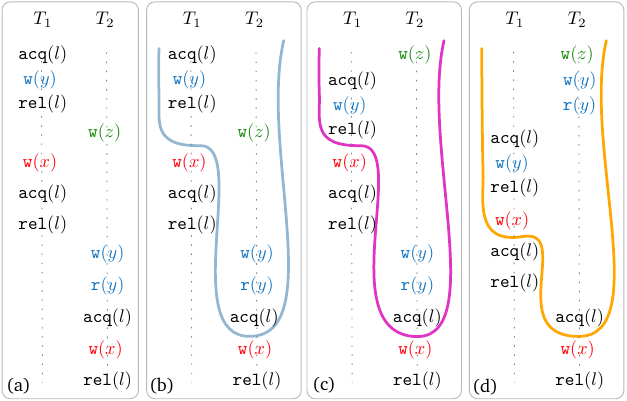}
		\caption{Prefix reasoning}
		\figlabel{exp-prefixes}
	\end{subfigure}
	\caption{Examples of data races detected by commutativity and prefix reasoning}
	\vspace{-15pt}
\end{figure}

\begin{example}
	Here we provide two example executions to illustrate and compare commutativity and prefix reasoning.
	First consider the three red $\wt(x)$ events in \figref{exp-commutativity}.
	The pair $(\ev{T_1, \textcolor{red}{\wt(x)}}, \ev{T_2, \textcolor{red}{\wt(x)}})$ is a $\maz$-race as the first critical section of $T_1$ is completely independent with $\ev{T_2, \textcolor{red}{\wt(x)}}$.
	The pair $(\ev{T_2, \textcolor{red}{\wt(x)}}, \ev{T_3, \textcolor{red}{\wt(x)}})$ is a $\grains$-race deduced by the commutativity of two blue grains $g_1$ and $g_2$.
	However, the pair $(\ev{T_1, \textcolor{red}{\wt(x)}}, \ev{T_3, \textcolor{red}{\wt(x)}})$ can only be detected by a scattered grain $g_3$ and a contiguous grain $g_2$.
	Notably, all the three races are prefix-races by the prefix marked as pink.
	Then, in \figref{exp-prefixes}, we show that the pair $(\ev{T_1, \textcolor{red}{\wt(x)}}, \ev{T_2, \textcolor{red}{\wt(x)}})$ is a prefix-race (thus also a $\seqp$, $\confp$, and $\syncp$ race) by a $\seqp$-prefix in (b), a $\confp$-prefix in (c), and a $\syncp$ race in (d) respectively.
	We also note that this is not $\maz$(or $\grains$, $\scatteredgrains$)-race due to the dependency between the second critical section in $T_1$ and $\ev{T_2, \acq(l)}$.
\end{example}


\section{Combining Commutativity and Prefix Reasoning}
\seclabel{combining}

While both commutativity and prefix reasoning offer the promise
of monitorability, i.e., streaming constant space algorithms,
the predictive power they offer tends to be limited.
In this work, we investigate how to enhance the power of these two
reasoning schemes. 
In particular, can we combine the two and arrive at a more powerful
predictive data race detection algorithm?
The focus of this section is to discuss a number of ways to combine commutativity and 
prefix reasoning that seem intuitive but  simply do not work in 
the sense that no additional expressive power in race detection can be gained from the combination. 



\myparagraph{A vanilla combination}{
A simple approach to this combination can be to design an algorithm
that simply checks for both types of races simultaneously.
That is, we design an algorithm $\mathcal{A}$ that,
on input $\tr$ returns true iff either $\tr$ has a $C$-race 
(for some $C \in \set{\maz, \grains, \scatteredgrains}$)
or if $\tr$ has a prefix-race.
As~\thmref{prefix-subsumes-commutativity}, suggests, however, this will yield no new expressive power since
prefix reasoning alone can discover the entire set of races.
}

\myparagraph{Generalizing prefixes using commutativity}
Recall the race illustrated in \figref{syncp-miss}(a), that is missed by prefix reasoning. 
Indeed, to witness the race on the two $\wt(x)$ events,
the only viable prefix is the one that contains exactly
the set of events $\set{\ev{T_1, \acq(l)}, \ev{T_2, \acq(l)}, \ev{T_2, \rel(l)}}$.
Unfortunately though, the only $\seqp$-prefix comprising of
exactly these three events cannot be well-formed since the earlier critical
section on lock $l$ must be unmatched and thus
overlap with the later critical section in this prefix.
Nevertheless, this example does 
suggest a different  approach to a combination of prefix and commutativity
reasoning for data race prediction --- generalization
of the space of prefixes by augmenting them 
through commutativity reasoning. 
In our example run \figref{syncp-miss}(a),
reordering the lock block of thread $T_2$ 
to execute before the that of thread $T_1$ would 
possibly be an instance of such a generalization.

As a first step to formalize this idea, 
we define the class of prefix races that can be obtained by 
generalizing prefixes with the different commutativity granularities 
we discussed in \secref{recap-commutativity}.
\begin{definition}[Commutativity-augmented-prefix races.]
Let $C \in \set{\maz, \grains, \scatteredgrains}$ be a choice of
commutativity granularity. 
For a run $\tr$,
we say that a run $\rho$ is a $C$-augmented prefix of $\tr$ if 
there is a $\seqp$-prefix $\rho'$ of
$\tr$ such that $\rho \in \eqcl{\rho'}{C}$.
Further, a pair $(e_1, e_2)$ of conflicting events
of $\tr$ is a $C$-augmented prefix-race if there is a $C$-augmented
prefix $\rho$ of $\tr$ in which both $e_1$ and $e_2$ are $\tr$-enabled.
\end{definition}

We are now sufficiently equipped to ask --- (1) 
\emph{how large the class of commutativity-augmented prefix races are}, and (2) 
\emph{how efficiently such data races can be detected?}. 
In light of answering question (2),
our focus is intentionally limited to
three types of commutativity reasonings (outlined in \secref{recap-commutativity})
for which known efficient algorithms exist.
Unfortunately, unlike the intuition from failed 
instances like the example in \figref{syncp-miss}(a), the answer to 
(1) is immediately discouraging under these constraints.
 That is, augmenting any of the prefix classes with 
any of the three types of commutativity reasoning discussed 
in~\secref{recap-commutativity} does not add any extra predictive power for data race prediction.

Indeed, this follows straightforwardly from the
definition when $C = \maz$ that of $\seqp$-prefixes, i.e., 
any $\maz$-augmented prefix is just
a $\confp$-prefix and thus $\maz$-augmented prefix-races are simply $\confp$-races,
which are also prefix-races.
Here, we show that this observation extends to all other commutativity granularities.
This is because commutations fundamentally do not change 
enabled-ness --- a reordering $\rho$ obtained by commuting a $\seqp$-prefix 
$\rho'$ has the same set of events enabled as $\rho'$.
That is, we have:

\begin{theorem}
\thmlabel{augmentating-prefixes-does-not-help}
Let $C \in \set{\maz, \grains, \scatteredgrains}$.
Let $\tr$ be an execution and $e_1, e_2$ be conflicting events in $\tr$.
The pair $(e_1, e_2)$ is an $C$-augmented prefix race
of $\tr$ iff $(e_1, e_2)$ is a prefix-race of $\tr$.
\end{theorem}


\begin{remark}
A careful reader may observe that, in principle, 
one can further extend the definition of commutativity augmented prefix races
by generalizing the class of prefixes beyond $\seqp$ prefixes
to include $\confp$ or $\syncp$ prefixes, 
i.e., by defining a $C$-augmented $P$-race
($C \in \set{\maz, \grains, \scatteredgrains}$, $P \in \set{\seqp, \confp, \syncp}$),
which is a pair of events $(e_1, e_2)$
in execution $\tr$ for which there is a $P$-prefix $\rho'$ of $\tr$
and a $\rho \in [\rho']_{C}$ such that both $e_1$ ad $e_2$ are $\tr$-enabled in $\rho$.
Unfortunately, the observation in \thmref{augmentating-prefixes-does-not-help}
extends to this class of races as well, for the same reasons.
That is, every $C$-augmented $P$-race is a prefix race, for
each $C \in \set{\maz, \grains, \scatteredgrains}$ and for each $P \in \set{\seqp, \confp, \syncp}$.
\end{remark}


In \secref{stratification}, we show that a more comprehensive combination
that enhances the class of predictive reasoning by
using commutativity reasoning beyond the prefix identified by one of the previously discussed classes.

\section{Stratifying Prefix and Commutativity Reasonings}
\seclabel{stratification}

Recall the example run in Figure \ref{fig:syncp-miss}(b), and consider the prefix marked with the red curve. We argued that the $\wt(x)$ of $T_3$  is not enabled after this prefix, and therefore the prefix cannot witness the race. If we consider the $\seqp$-prefix marked by the green curve, however, all the remaining events including the blue $\rd(y)$ and the red $\wt(x)$ events of thread $T_3$ are enabled as a sequence after this prefix. Then, in the remaining enabled sequence, the two red $\wt(x)$s are simply an example of $\maz$-race. This motivates the key concept for considering the races after a prefix to be the sequence of events that are {\em executable} following a given $\seqp$-prefix, in which we can {\em predict} a race.


\begin{definition}[Enabled Sequence of Events]
    Given a correct reordering $\rho$ of a run $\sigma$. 
    A subsequence $\tau$ of $\sigma$ is enabled after $\rho$ if 
    \begin{enumerate*}
        \item $\events{\rho} \cap \events{\tau} = \emptyset$,
        \item $\events{\rho}\cup\events{\tau}$ is $\po{\sigma}$-closed,
        \item $\rho\cdot\tau$ is well-formed, and
        \item for every read event $e\in\events{\tau}$ such that $\rf{\rho\cdot\tau}(e) \neq \rf{\sigma}$, it is the last event of its thread in $\tau$, i.e., for all $e'\in\events{\rho\cdot\tau}$ such that $\ThreadOf{e'} = \ThreadOf{e}$, we have $(e', e)\in\po{\sigma}$.
    \end{enumerate*}
\end{definition}

\begin{definition}[$\seqp$-suffix of $\rho$]
Given a $\seqp$-prefix $\rho$ of a program run $\sigma$, 
$\tau$ is a $\seqp$-suffix of $\rho$ if $\tau$ is enabled  after $\rho$ and
 $\events{}(\tau)$ appear precisely in the same order in $\tau$ as they do in $\sigma$.
We call $\rho$ an {\em enabling} prefix of $\tau$. We call $\tau$ a {\em maximal $\seqp$-suffix} of $\rho$, if it is a $\seqp$-suffix of $\rho$ and not a subsequence of any other $\seqp$-suffix of $\rho$.
\end{definition}

\begin{wrapfigure}[7]{r}{0.18\textwidth}
\vspace{-10pt}
\includegraphics[scale=0.9]{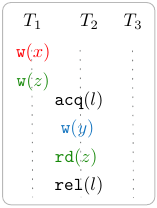}
\vspace{-15pt}
\end{wrapfigure}
Observe that after a prefix $\rho$, one can keep including the remaining events from $\sigma$ as long as they are executable up to the set of already included events, and as such the concept of a maximal $\seqp$-suffix of $\rho$ is well-defined, but maximal suffixes are not unique. For the $\seqp$-prefix marked by the red curve in Figure \ref{fig:syncp-miss}(b), the maximal $\seqp$-suffix is illustrated on the right. The $\wt(x)$ event of $T_3$ cannot appear, but the rest of events can be executed in order. Naturally, any prefix of the run illustrated on the right is also a $\seqp$-suffix (although no longer maximal).

The key property of $\seqp$-suffixes is that one can treat them as standalone runs, predict races in them, and have the guarantee that any predicted races are also sound for the original run. For example, there is a race between the two green events above, since $\wt(z)$ commutes against the next two events. The reader can verify that the same race exists in the original run in Figure \ref{fig:syncp-miss}(b).


\begin{theorem}\label{thm:ssound}
Let $\tau$ be a $\seqp$-suffix of a program run $\sigma$. If events $e_1, e_2 \in \events{\tau}$ form a (predictable) race in $\tau$, then they form a (predictable) race in $\sigma$. 
\end{theorem}

Any $\seqp$-suffix $\tau$ of $\sigma$ is induced by a $\seqp$-prefix $\rho$. 
By definition, there exist an execution $\sigma'$ of the program in which $\rho$ appears (in the same order as the original program run) followed by $\tau$, also with events appearing in the same order; that is $\sigma' = \rho \tau$ is a feasible execution of the same program. If we know that a predictable race in $\tau$ is a predictable in $\rho \tau$, then we know this race is a valid race for the program. This is an implication of the following generic lemma about predictable races:

\begin{restatable}{lemma}{Monotonicity}
    \label{lem:monotonicity}
Let $\sigma = \alpha\beta$ be a program run. If events $e_1, e_2 \in \events{\beta}$ form a (predictable) race in $\beta$, then they form a (predictable) race in $\sigma$. 
\end{restatable}
In a sense, $\seqp$-suffixes bring a power of localizing the search for a race to a subsequence (not necessarily contiguous) of the original run. Inspired by this, we define two classes of races for which efficient algorithmic solutions exists. We compare the expressiveness of these classes of races against each other and the baseline $\seqp$-races, and present an algorithm for the most expressive class in the next section.

\subsection{Maximal $\seqp$-Suffix Reasoning}\seclabel{suffix}
The first class of races focuses on the maximal $\seqp$-suffixes and the races that can be predicted in them using commutativity.

\begin{definition}[Maximal Suffix $C$-Race]
A pair of events $e_1$ and $e_2$ from a program run $\sigma$ form a maximal-suffix race iff there exists a $\seqp$-prefix $\rho$ and a maximal $\seqp$-suffix $\tau$ of $\rho$ such that $(e_1,e_2)$ form a $C$-race in $\tau$ for $C \in \set{\maz, \grains, \scatteredgrains}$.
\end{definition}

One can obviously predict a race in $\tau$ using a more sophisticated/expensive scheme, 
but the performance of any such scheme could be unreasonably poor. 
First, let us remark on a simple connection between these races and standard commutativity races.
\begin{proposition}
For any program run and any $C \in \set{\maz, \grains, \scatteredgrains}$, the set of $C$-races is strictly contained in the set of  maximal suffix $C$-races.
\end{proposition}
This is the consequence of the simple fact that for a given run $\sigma$, the maximal $\seqp$-suffix of an {\em empty} prefix is  $\sigma$ itself. Hence, the set of $C$-races is already subsumed by the set of maximal suffix $C$-races with leaving the choice of the $\seqp$-prefix to be empty. 
Unfortunately, and rather surprisingly, we cannot make a similar claim about $\seqp$-prefix races.

\begin{proposition}
The set of maximal $\seqp$-suffix $C$-races of a program run is not generally comparable with the set of its $\seqp$-prefix races.
\end{proposition}
\begin{wrapfigure}[10]{r}{0.18\textwidth}
\vspace{-15pt}
\includegraphics[scale=0.9]{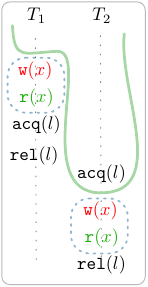}
\vspace{-15pt}
\end{wrapfigure}
The races between the pair of red $\wt(x)$ events in Figures \ref{fig:syncp-miss}(a) and \ref{fig:syncp-miss}(b) are both maximal $\seqp$-suffix $\maz$-races but not $\seqp$-prefix races. The same is true for the race between the pair of blue $\wt(y)$ events in Figure \ref{fig:grains}. 

For an example of a $\seqp$-prefix race that is not a maximal $\seqp$-suffix $C$-race (for any choice of $C$), consider the run illustrated on the right. Modulo the addition of the green $\rd(x)$ events, this run is identical to the one in Figure \ref{fig:syncp}, for which we argued in Section \ref{sec:intro} that prefix reasoning detects the race between the two red $\wt(x)$ events. The addition of the $\rd(x)$ event has no impact on that reasoning, and the race is still predictable using the marked prefix. There is only one maximal $\seqp$-suffix for the marked $\seqp$-prefix:
\begin{align*}
\tau_1 =&\ev{T_1,\wt(x)} \ev{T_1,\rd(x)} \ev{T_2,\wt(x)}\ev{T_2,\rd(x)} \ev{T_2,\rel(l)}
\end{align*}
The existence of $\ev{T_1,\rd(x)}$ event prevents us to argue using event commutativity that the two $\wt(x)$ events form a race in $\tau_2$. Grain commutativity can reorder the marked grains, but that does make the pair of $\wt(x)$ events adjacent. Scattered grains do not contribute any new correct reorderings. Therefore, this $\seqp$-prefix race is not a maximal $\seqp$-suffix $C$-race, for any notion of commutativity $C$.  
This, rather counterintuitive, fact that prefix reasoning gains some power, by simply isolating a pair of events enabled at the boundary, which is lost in maximal-suffix races motivates the class of races in the next section. 

\subsection{Granular Prefix Reasoning}\seclabel{granular}
A {\em grain} $g$ is a contiguous subword. In \cite{FarzanMathurPOPL2024}, grains are used as a way of gaining more commutativity when individual events do not commute, a sequence of events may commute against another. The key idea here is that grains can be helpful in prefix reasoning {\em as well}. Rather than focusing on single events in the boundary of a prefix forming a race, one can infer two grains to be racy in the same sense, and then discover a concrete race between a pair of events inside the two grains using commutativity reasoning as described in the previous section. 
The two grains $g_1$ and $g_2$ marked in Figure \ref{fig:grains} are enabled after the marked $\seqp$-prefix. Note that we cannot take the entire set of events in thread $T_1$ as a grain, because this would violate the contiguity requirement for grains. 


\begin{definition}[Granular Prefix $C$-Race]
\deflabel{grseqp}
A pair of events $e_1$ and $e_2$ from a program run $\sigma$ form a granular prefix race iff there exists a $\seqp$-prefix $\rho$, two (not necessarily distinct) maximal $\seqp$-suffixes $\tau_1$ and $\tau_2$ of $\rho$, and two grains $g_1$ of $\tau_1$ and $g_2$ of $\tau_2$, such that
\begin{itemize}[nosep]
\item $g_1$ include $e_1$ and $g_2$ includes $e_2$, and
\item $g_1g_2$ is enabled after $\rho$, and
\item $(e_1,e_2)$ is a $C$-race in $g_1g_2$, for $C \in \set{\maz, \grains, \scatteredgrains}$.
\end{itemize}
\end{definition}

We call a granular prefix $\maz$-race, for short, a $\grconfp$ race.

\begin{theorem}
\thmlabel{grconfp-subsumes-syncp}
For any program run and for all $C \in \set{\maz, \grains, \scatteredgrains}$, the set of granular prefix $C$-races strictly contains the set of $\seqp$-prefix races.
\end{theorem}
A grain can comprise a single event, and as such, it is straightforward why every $\seqp$-prefix race is also a granular race. For the strict inclusion, recall the example run in Figure \ref{fig:grains}, which incidentally also serves as an example of a granular prefix race that is not maximal-suffix races.

\begin{proposition}
\thmlabel{max-suffix-race-subsumes-syncp}
In any program run and for all $C \in \set{\maz, \grains, \scatteredgrains}$, the set of granular prefix $C$-races strictly contains the set of maximal suffix $C$-races.
\end{proposition}
If $e_1$ and $e_2$ form a maximal suffix $C$-race, witnessed by a maximal suffix $\tau$, then one can beak the maximal suffix $\tau$ into two grains $g_1$ and $g_2$ such that $\tau = g_1g_2$, by splitting $\tau$ right after the first of $e_1$ and $e_2$ appears. Using these grains $g_1$ and $g_2$, one can predict the same race by definition.

Grain and scattered-grain commutativity are strictly more expressive than event-based commutativity and hence the natural choice would be to use them in granular reasoning to discover even more races. Theoretically, there is no obstacle to this, since the use of all three notions would yield a constant-space algorithm. Practically, however, things are a bit different. 
It is well-understood that there is a tension between expressiveness and efficiency when it comes to the class of predictive race detection algorithms. In \cite{Mathur21}, it was observed that a constant space algorithm for prefix reasoning alone can behave poorly in practice. Intuitively, think of an algorithm as guessing all possible prefixes, which are maintained as a constant-bounded set of summaries. This constant yet large enough state space has to be carefully maintained every time a new event is processed by the algorithm, and the price of this maintenance over millions of events does not yield a practically efficient algorithm. It is therefore very important for any additional reasoning to be lightweight and very fast. This is the case for event commutativity, which yields a deterministic algorithm, but not grain commutativity which involves further guessing and therefore a blow up of the state space.


\section{Algorithm}
\seclabel{algo}

Here, we show that there is a \emph{constant space} streaming algorithm 
for detecting $\grconfp$-races.
In \secref{alg-overview}, we first highlight the key observations behind our algorithm: 
an automata-theoretic algorithm 
that simulates a nondeterministic finite-state automaton (NFA) $\aut_{\grconfp}$ 
and an effective \emph{antichain optimization} 
that improve the space usage when naively determinizing $\aut_{\grconfp}$.
Then, in \secref{alg-detail}, we present the NFA construction by specifying its state structure and transition function. 
We also discuss a concrete antichain optimization based on the proposed NFA.

\subsection{Overview of the Automata-Theoretic Algorithm}
\seclabel{alg-overview}

We first present a high-level overview of the automata-theoretic algorithm for detecting $\grconfp$-races, i.e., a nondeterministic finite automaton $\aut_{\grconfp}$ accepting 
$L_{\grconfp} = \setpred{\tr \in \labs^*}{\tr \text{ has a } \grconfp\text{-race}}$.
Recall that the definition of $\grconfp$-race follows an existential quantification over subsequences of events in $\tr$.
The automaton $\aut_{\grconfp}$ is designed to employ a "guess-summarize-check" strategy that precisely simulates the definition in a streaming fashion.
The automaton $\aut_{\grconfp}$ first guesses a prefix $\rho$ of $\tr$, two grains $g_1$ and $g_2$, and two conflicting events $e_1$ and $e_2$.
It then summarizes the information about the guessed subsequences and checks if they satisfy the conditions.

Guessing is straightforward thanks to nondeterminism.
When scanning the execution $\tr$, the automaton $\aut_{\grconfp}$ can nondeterministically guess the role of the processed event.
The challenging part is to design the state structure to summarize the guessed subsequences such that the summarized information 
\begin{enumerate*}
    \item is sufficient to check the conditions in \defref{grseqp}, and 
    \item takes constant space.
\end{enumerate*}
Recall that the conditions in \defref{grseqp} consists of four parts:
\begin{enumerate*}[label=(\roman*)]
    \item checking if $e_1$ is included in $g_1$, $e_2$ is included in $g_2$ (\underline{\chk-grains}),
    \item checking if $\rho$ is a \seqp-prefix (\underline{\chk-seqp}),
    \item checking if $g_1$ and $g_2$ are enabled after $\rho$ (\underline{\chk-enabled}) and
    \item checking if, together, they witness $(e_1, e_2)$ to be in $\maz$-race (\underline{\chk-race}).
\end{enumerate*}
The first condition can be ensured by only guessing the events that are included in the grains.
For \underline{\chk-seqp} and \underline{\chk-enabled}, a crucial observation is that both checking can be done incrementally.
When $\aut_{\grconfp}$ guesses an event $e$ as the next event in $\rho$ (or $g_1$ or $g_2$), it can perform the enabledness checking against the summarized information of the previous guessed $\rho$ (or $g_1$ or $g_2$).
Once $\rho\circ e$ (or $g_1\circ e$ or $g_2\circ e$) is checked to be a \seqp-prefix (or enabled), the automaton can update the summarized information accordingly.
We will give the detailed construction of the summarized information in \secref{alg-detail},
but the key idea that provides a constant-space result is to track a set of threads, memory location, and locks such that accesses to them are not enabled.
The last condition \underline{\chk-race} that checks if two events are a $\maz$-race can be solved by a scalable automata-theoretic algorithm following classical results in trace theory~\cite{diekert1995book}, whose details are presented in \secref{alg-detail}.
Finally, the state is a tuple that contains the summarized information to perform each checking.
The state is an accepting state (the execution has a $\grconfp$-race) if all the conditions are satisfied.

To obtain a streaming, constant-space algorithm from the automaton $\aut_{\grconfp}$, a naive way is to determinize the automaton on-the-fly.
However, the naive determinization algorithm may suffer from an exponential dependence on parameters like $|\threads|, |\mems|$, and $|\locks|$ and is expected to have poor performance when any of these parameters is moderately large.
Indeed, for this reason, the constant space algorithm for \seqp-races was observed to scale very poorly as compared to the linear space algorithm for \syncp-races~\cite{Ang2024CAV}.
In this paper we take inspiration from the \emph{antichain optimization}~\cite{anitChain2006}, 
previously proposed in the context of the \emph{universality problem} for NFAs~\cite{anitChain2006}, for our setting of solving the membership problem against our proposed NFA $\aut_{\grconfp}$.
Indeed, our optimization also applies to,
and is effective for the constant space algorithm for \seqp-races (see \secref{exp-result}).

Intuitively, we identify a \emph{subsumption} partial order $\subsumes$
on the states of our NFA such that
for every two states $p, p'$, whenever $p \subsumes p'$, 
then for each word $w \in \labs^*$,
if there is an accepting run of $w$ starting from $p'$, then there is one from $p$ as well.
In turn, this means that, when running checking for the membership of $\tr$ on the NFA
using a typical on-the-fly subset-construction algorithm, 
it suffices to instead track only the states that are \emph{minimal} 
according to the partial order $\subsumes$.
Once a concrete definition of $\subsumes$ is defined, an algorithm
for membership can essentially simulate the NFA, while removing
new non-minimal states at each step in the algorithm by comparing
all pair of states according to $\subsumes$.
We present the concrete subsumption partial order and the antichain optimization in \secref{alg-detail}.
Here, we provide a taste of the optimization by an example.
Consider two guessing of \seqp-prefix $\rho_1$ and $\rho_2$ (where there is no guessed $g_1$ and $g_2$) in the middle of a run.
$\rho_1$ and $\rho_2$ are the same except that $\rho_1$ include one more write event.
In this case, $\rho_2$ is subsumed by $\rho_1$ ($\rho_1 \subsumes \rho_2$) since the set of the enabled events after $\rho_1$ is larger than the enabled set after $\rho_2$.
Specifically, there might be a read event on the same variable that is enabled after $\rho_1$ but not after $\rho_2$.
Therefore, the antichain optimization can remove $\rho_2$ from the state space.

\subsection{Construction of $\aut_{\grconfp}$ and Antichain Optimization}
\seclabel{alg-detail}

Following the overview section, we first present how to perform the three essential checkings (\underline{\chk-\seqp}, \underline{\chk-enabled}, and \underline{\chk-race}) required in the definition of $\grconfp$-races using constant-space summarized information.
We then introduce an antichain optimization for $\aut_{\grconfp}$, which extends our newly proposed antichain optimization for
$\seqp$-race detection algorithms.
This extension is particularly relevant as an automata-theoretic $\seqp$-race detection algorithm can be conceptualized as a simplified version of $\aut_{\grconfp}$.

At the top level, each state of $\aut_{\grconfp}$ is represented as a tuple $p = (e_1, p_\rho, p_{g_1}, p_{g_2}, \aftset_{e_1}, \acr{phase})$, where the first four components provide a summary of the guessed subsequences $e_1$, $\rho$, $g_1$, and $g_2$, the component $\aftset_{e_1}$ summarizes the events that are ``dependent'' with $e_1$, and $\acr{phase}$ tracks the ``arrangement'' of the current event
and can be one of (\acr{Before}) ``the current event occurs before $g_1$'', 
(\acr{Inside$_1$}) ``the current event occurs inside $g_1$'', 
(\acr{Between}) ``the current event occurs after $g_1$ but before $g_2$'',
or 
(\acr{Inside$_2$}) ``the current event occurs inside the grain $g_2$''.
When processing a new event $e$, $\aut_{\grconfp}$ nondeterministically decides to stay in the same phase or move to the next phase, and makes corresponding role guessing (including $e$ in $\rho$ or not, etc).

\myparagraph{Check if $\rho\circ e$ is a \seqp-prefix}
Assume that $p_\rho$ is a summary of the guessed \seqp-prefix $\rho$.
We first provide the definition of $p_\rho$ and then show how to check and update $p_\rho$ when guessing the next event $e$.
The summary $p_\rho$ is a tuple 
\[p_\rho = (\acr{ExclThreads}, \acr{OpenLcks}, \acr{ExclLastWrs}),\]
where 
\begin{enumerate*}
    \item $\acr{ExclThreads}$ is a set of threads that have events guessed to be excluded from $\rho$,
    \item $\acr{OpenLcks}$ is a set of locks whose acquire is included in $\rho$, but the release is excluded, and 
    \item $\acr{ExclLastWrs}$ is a set of memory locations whose latest write event is guessed to be excluded from $\rho$.
\end{enumerate*}
$p_\rho$ is of constant size due to the assumption that $|\threads|$, $|\mems|$, and $\locks$ are constant.
In phase $\acr{Before}$ or $\acr{Between}$, if the processed event $e$ is guessed to be next in $\rho$, we check the following conditions:
\begin{enumerate*}[label=(\roman*)]
    \item if $\ThreadOf{e}\not\in p_\rho.\acr{ExclThreads}$,
    \item if $l\not\in p_\rho.\acr{OpenLcks}$ when $e = \acq(l)$, and
    \item if $x\not\in p_\rho.\acr{ExclLastWrs}$ when $e = \rd(x)$.
\end{enumerate*}
These checkings ensure that $e$ is enabled after $\rho$ and that $\rho\circ e$ is still a \seqp-prefix.

\myparagraph{Check if $e$ is enabled after $\rho\circ g_1$}
We also assume that $p_{g_1}$ is a summary of the guessed $g_1$ and $\rho$ is a summary of the guessed \seqp-prefix $\rho$, where $g_1$ is enabled after $\rho$.
Similarly, $p_{g_1}$ is a tuple
\[p_{g_1} = (\acr{InclThreads}, \acr{InclAcqs}, \acr{OrphanRels}, \acr{InclWrs}, \acr{OrphanRds}),\]
where 
\begin{enumerate*}
    \item $\acr{InclThreads}$ is a set of threads that have events guessed to be included in $g_1$,
    \item $\acr{InclAcqs}$ is a set of locks whose acquire is included in $g_1$,
    \item $\acr{OrphanRels}$ is a set of locks whose release is included in $g_1$ but the corresponding acquire event is not,
    \item $\acr{InclWrs}$ is a set of memory locations whose write event is included in $g_1$, and
    \item $\acr{OrphanRds}$ is a set of memory locations whose read event is included in $g_1$ but the corresponding write event is not.
\end{enumerate*}
In phase $\acr{Inside_{1}}$, if the processed event $e$ is guessed to be next in $g_1$, we check the following conditions:
\begin{enumerate*}[label=(\roman*)]
    \item if $\ThreadOf{e}\in p_{g_1}.\acr{InclThreads}$ or $\ThreadOf{e}\not\in p_\rho.\acr{ExclThreads}$,
    \item if $l\in p_{g_1}.\acr{InclAcqs}$ or $l\not\in p_\rho.\acr{OpenLcks}$ when $e = \acq(l)$,
    \item if $x\in p_{g_1}.\acr{InclWrs}$ or $x\not\in p_\rho.\acr{ExclLastWrs}$ when $e = \rd(x)$.
\end{enumerate*}

However, in phase $\acr{Between}$, an extension of $\rho$ with $f$ might disable the existing $g_1$, additional checkings are necessary:
\begin{enumerate*}[label=(\roman*)]
    \item if $l\not\in p_{g_1}.\acr{OrphanRels}$ when $e = \acq(l)$, and
    \item if $x\not\in p_{g_1}.\acr{OrphanRds}$ when $e = \wt(x)$.
\end{enumerate*}
These combined checkings guarantee an invariant that $g_1$ is always enabled after $\rho$.
A similar checking for $g_2$ follows the same spirit, and we omit the details here.

\myparagraph{Check if $(e_1, e_2)$ is a $\maz$-race in $g_1\circ g_2$} 
From the classical results in trace theory, $(e_1, e_2)$ is a $\maz$-race if $e_2$ does not belong to the set containing events that are transitively dependent on $e_1$, denoted by $\acr{AftEvent}_{e_1}$.
Formally speaking, $\acr{AftEvent}_{e_1}$ is the smallest set that includes $e_1$ and satisfies 
\[\forall\ e\in g_1\circ g_2,\ \exists f \in \acr{AftEvent}_{e_1} \text{ such that } f \trord{g_1\circ g_2} e \land (f, e)\not\in \indrel \Longrightarrow e \in \acr{AftEvent}_{e_1}.\]
It is easy to see that $\acr{AftEvent}_{e_1}$ can be maintained incrementally, however, the size of $\acr{AftEvent}_{e_1}$ can be as large as the length of the execution, which is not constant.
To address this issue, a crucial observation is that the independent relation $\indrel$ only relies on the label of events (thread identifier, operation, and memory location (or lock)).
Therefore, we can maintain a set of labels $\acr{AftSet}_{e_1}$ that contains the labels of events in $\acr{AftEvent}_{e_1}$, whose size is upper bounded by $2\times|\threads| \times (|\mems| + |\locks|)$.
In fact, it suffices to decouple the labels and only store a set of dependent threads, read memory locations, write memory locations, and accessed locks, which reduces to size to $|\threads| + 2\times|\mems| + |\locks|$.

\myparagraph{Subsumption relation for \seqp-races}
We can easily transform the above construction of the $\aut_{\grconfp}$ to an NFA $\aut_{\seqp}$ for detecting \seqp-races
by limiting the size of $g_1$ and $g_2$ to one.
Now we first define a subsumption relation $\subsumes_\seqp$ on the states of $\aut_{\seqp}$,
then extend it to $\aut_{\grconfp}$.
A state of $\aut_{\seqp}$ is a tuple $p = (e_1, p_\rho)$.
We define the subsumption relation $\subsumes_\seqp$ as follows:
\[p \subsumes_\seqp p'
    \text{ iff }
    p.e_1 = p'.e_1 \land p_\rho \subsumes p'_\rho,
\]
where the subsumption relation between the summaries $p_\rho$ is a component-wise subset relation, i.e., $p_\rho \subsumes p'_\rho$ iff $p_\rho.\acr{ExclThreads} \subseteq p'_\rho.\acr{ExclThreads}$, etc. 
Intuitively, $p \subsumes_\seqp p'$ when they track the same $e_1$ and $p.p_\rho$ enables more events than $p'.p_\rho$.
The following justifies why the subsumption relation is crafted as above:
on these states is crafted such that
if state $p'$ can transition to state $q'$,
then a state $p$, satisfying $p \subsumes_\seqp p'$,
can also transition to some $q$ satisfying $q \subsumes_\seqp q'$.

\begin{lemma}
The set of accepting states of the NFA $\aut_\seqp$ is downward closed w.r.t. $\subsumes_\seqp$.
Further, let $\tr$ be a run and let $p$ and $p'$ be states in $\aut_\seqp$ reached after reading $\tr$
that satisfy $p \subsumes_\seqp p'$.
Let $\sigma \in \labs$ and let $q'$ be a state such that $p'$ can transition to $q'$ on reading $\sigma$. 
Then there is a state $q$ such that $p$ can transition to $q$ and $q \subsumes_\seqp q'$.
\end{lemma}


\myparagraph{Subsumption relation for $\aut_{\grconfp}$-races}
A very thorough subsumption relation can, in principle, be defined
and implemented even for the more elaborate NFA $\aut_{\grconfp}$, such as an additional subset relation on $\aftset_{e_1}$.
However, an overly complex definition hinders efficiency since isolating the
set of minimal states may become expensive.
In view of this, we opt for a simple subsumption relation $\subsumes_{\grconfp}$
that is an extension of the relation $\subsumes_\seqp$ we defined above:
\[
    p \subsumes_{\grconfp} p'
    \text{ iff }
    p_\rho \subsumes p'_\rho \land \pi_i (p) = \pi_i (p') \text{ for }i \in \{1, 3, 4, 5, 6\},
\]
where $\pi_i$ is the projection function that maps each state to its $i$-th component.
As before, soundness follows because each step in the automaton
is monotonic w.r.t this relation:
\begin{lemma}
The set of accepting states of the NFA $\aut_{\grconfp}$ is downward closed w.r.t. $\subsumes_{\grconfp}$.
Further, let $\tr$ be a run and let $p$ and $p'$ be states in $\aut_{\grconfp}$ reached after reading $\tr$
that satisfy $p \subsumes_{\grconfp} p'$.
Let $\sigma \in \labs$ and let $q'$ be a state such that $p'$ can transition to $q'$ on reading $\sigma$. 
Then there is a state $q$ such that $p$ can transition to $q$ and $q \subsumes_{\grconfp} q'$.
\end{lemma}

   


\section{Experimental Results}
\seclabel{experiments}


In this section, we evaluate an optimized version of \seqp algorithm~\cite{Ang2024CAV} ($\optseqp$) and our main algorithm $\optgrconfp$ for prediction of $\grconfp$-races.
We first discuss some key implementation details for our algorithms (\secref{algo}), including how we determinize our theoretical 
automata-based algorithm effectively (\secref{exp-impl}), and our key optimizations. Then, in \secref{exp-result}, we present the result of our evaluation on a standard benchmark suite to measure the expressiveness and performance of our
algorithm and proposed optimizations.

\subsection{Implementation and Empirical Setup}
\seclabel{exp-impl}

\myparagraph{Implementations}
We implement two algorithms based on the discussion in \secref{alg-detail}:
\begin{enumerate*}
    \item an antichain-optimized algorithm \optseqp for prediction of \seqp-races~\cite{Ang2024CAV}, and
    \item the main algorithm for prediction of $\grconfp$-races.
\end{enumerate*}
Our implementation converts the conceptual NFA into a practical streaming algorithm by simulating the nondeterministic procedure --- for each incoming event, we track and update all possible states across all nondeterministic choices.
This naive algorithm can accumulate $O(2^{\poly(|\threads| + |\vals| + |\locks|)})$
states and is expected to not scale.
For $\optseqp$, we only apply the antichain-based optimization.
Now we discuss the additional optimizations and heuristics we apply to $\grconfp$ to achieve scalability.

\myparagraph{Complete Optimizations}
We use two completeness-preserving optimizations.
These are (1) the antichain-based subsumption discussed 
in \secref{alg-detail}, and (2) state partitioning that enables parallelism. 
Recall that each state tracks potential races on a fixed memory location $x$, 
and states for different memory locations are updated independently. 
Therefore, we can partition the state space according to
memory locations and update these partitions in parallel.
We use a $16$-core machine to parallelize the state space update.

\myparagraph{Heuristics}
We also use $3$ heuristic optimization which can be aggregated together:
\begin{enumerate}[nosep,topsep=0pt]
    \item \textbf{Limiting grain sizes.}
    This optimization, denoted by $\hSz^m$, imposes an upper bound $m$ on the length of grains. 
    This restriction is inspired by the observation that 
    larger grains are more likely to include elements that hinder the followup commutativity reasoning.

    \item \textbf{Limiting grain shapes.}
    This optimization limits the choices of $g_1$ to those that contain
    a single complete critical section and the choices of $g_2$ to be a singleton grain. 
    This retains one of the key benefits that makes
    $\grconfp$ more predictive than \seqp --- it
    can reason about reorderings that invert the order of critical sections,
    when the later critical section is in the prefix $\rho$,
    while the earlier one is in the grain $g_1$.
    We use the symbol $\hSz$ for this heuristic.

    \item \textbf{Evicting LRU states.} 
    The final optimization exploits the temporal locality of 
    data races --- the likelihood of two conflicting events to race,
    typically decreases as their temporal distance increases. 
    To put locality in action, we augment each state $q$ with a 
    timestamp recording when the component $e_q$ was added. 
    We then bound the state space by eliminating states 
    corresponding to the \emph{least recently used} $e_q$ components. It is denoted by $\hLRU^n$, for a bound $n$ on the number of states.
\end{enumerate}

We use $\optgrconfp$ to denote our
optimized algorithm and the 
notation $[-/\hSz^m, -/\hSh, -/\hLRU^n]$ for our
heuristic configuration 
where $-$ indicates the absence of the corresponding heuristic.

\myparagraph{Benchmarks}
We evaluate our algorithms against benchmarks from~\cite{Mathur21},
consisting of concurrent programs taken from standard benchmark suites and recent literature:
\begin{enumerate*}[label=(\roman*)]
    \item the IBM Contest suite~\cite{farchi2003concurrent},
    \item the DaCapo suite~\cite{blackburn2006dacapo},
    \item the Java Grande suite~\cite{sen2005detecting},
    \item the Software Infrastructure Repository suite~\cite{do2005supporting}, and
    \item others~\cite{Mathur21,rvpredict}.
\end{enumerate*}
To ensure a fair comparison, we generated a single trace for each benchmark using
{\sc RVPredict}~\cite{rvpredict} and performed all evaluations on the same trace. 

\myparagraph{Compared Methods}
We focus on our algorithms, \optseqp and $\optgrconfp$ under  different heuristics with state-of-the-art data race predictors, \shb~\cite{Mathur18}, \wcp~\cite{Kini17}, \syncp~\cite{Mathur21}, and \osr~\cite{OSR2024} 
on the same input trace $\tr$.
In terms of expressiveness, $\grconfp$ is strictly more powerful than $\syncp$, which strictly subsumes $\shb$, while $\osr$ and $\wcp$ are incomparable with the other three algorithms. 
For a fixed trace $\tr$, we report races as all events $e_2$ in $\tr$ that are detected as racy with a previous event $e_1$ where $e_1 \trord{\tr} e_2$. 
In addition to racy events, we also reported the number of buggy locations (lines) in the source code, since a single location can be counted several times as distinct racy events.
On all evaluations, we set a 3-hour timeout.
Our experiments were conducted on a 64-bit Linux machine with Java 19 and 128 GB heap space.

\begin{figure}[t]
    \centering
        \begin{subfigure}{0.45\textwidth}
            \centerline{\includegraphics[width=0.95\textwidth]{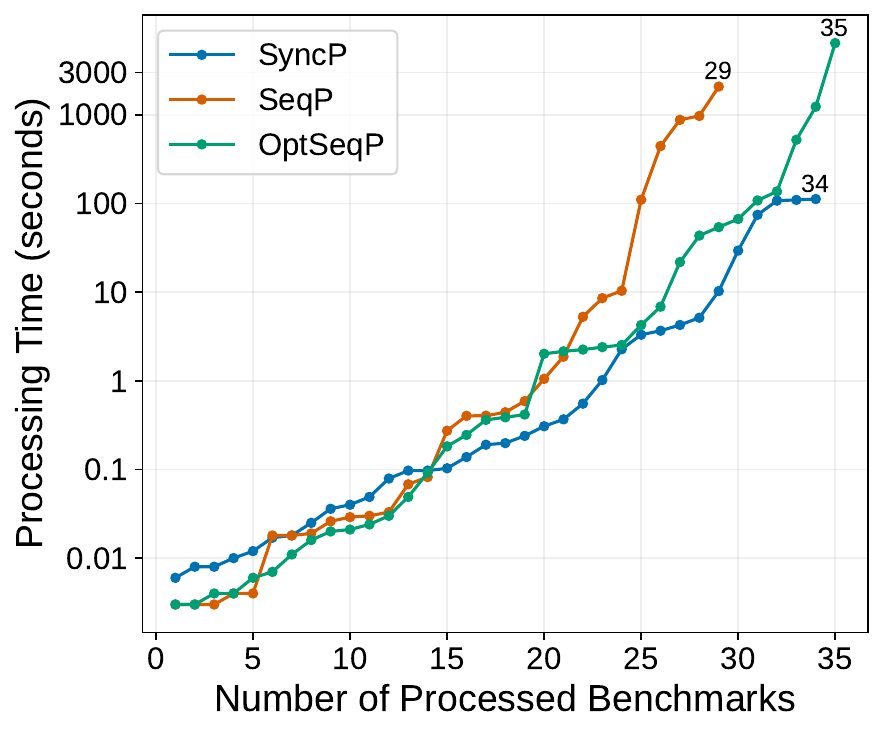}}
            \caption{Performance of \syncp, \seqp, and \optseqp}
            \figlabel{subsumption}
        \end{subfigure}
        \hfill
        \begin{subfigure}{0.45\textwidth}
            \centerline{\includegraphics[width=0.95\textwidth]{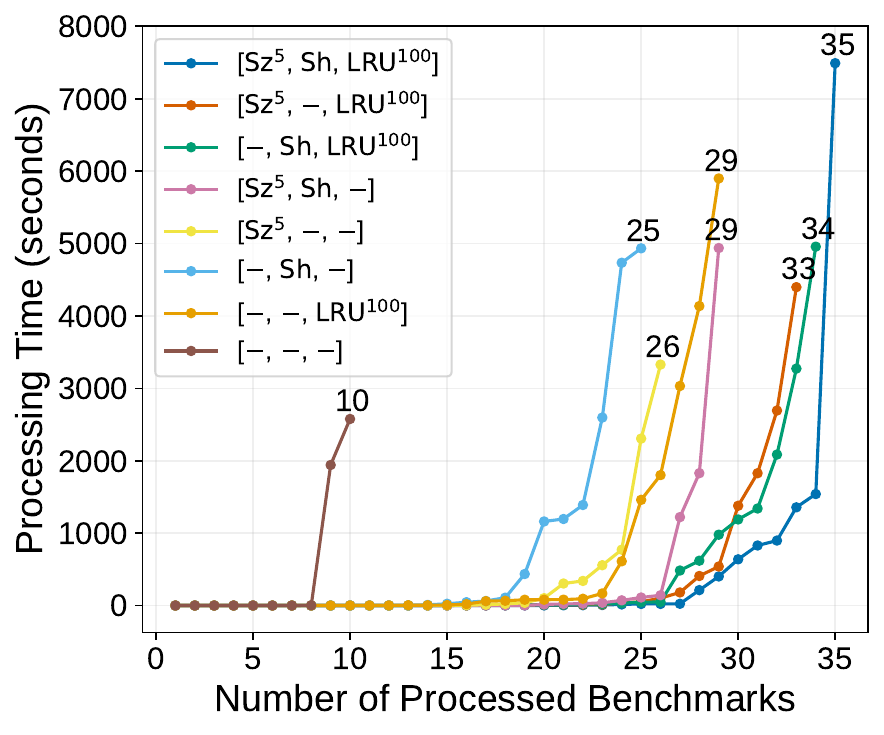}}
            \caption{Performance of $\grconfp$ and $\optgrconfp$}
            \figlabel{heuristic}
        \end{subfigure}
        \vspace{-10pt}
        \caption{Performance Comparison}
        \figlabel{linear}
        \vspace{-10pt}
\end{figure}

\subsection{Empirical Evaluation}
\seclabel{exp-result}

We designed our experiments to answer the following research questions:
\begin{description}[itemsep=0pt,topsep=0pt,leftmargin=0cm]
    \item[RQ1 (Expressiveness).] 
    How does $\grconfp$'s race predictive power compare to state-of-the-art algorithms~\cite{Kini17,Mathur18,OSR2024,Mathur21}? Particularly, can it predict both previously known and new races?
    
    \item[RQ2 (Performance).] 
    How well does $\grconfp$'s constant-space linear-time algorithm scale?
    
    \item[RQ3 (Effectiveness of optimizations and heuristics).] 
    How well does antichain optimization enhance the scalability? How well do our heuristics balance scalability and expressiveness?
\end{description}

\myparagraph{Effectiveness of antichain optimization}
We first address the first part of RQ3, examining the effectiveness of antichain optimization by comparing the performance of \seqp and $\optseqp$.
\figref{subsumption} illustrates the results using a logarithmic scale.  
For ``simpler'' benchmarks, \seqp and \optseqp show similar performance, as the original state-space is small. 
However, when processing ``harder'' benchmarks, \optseqp significantly outperforms \seqp, reducing processing time with a factor of 10 for an equivalent number of benchmarks. 
Consequently, \optseqp successfully processes 6 additional benchmarks within the three-hour timeout period. 
In addition, \optseqp also processes one more benchmark than \syncp, where \syncp, a linear-space algorithm, encounters an Out-of-Memory (OOM) error.
This performance gain is the result of our antichain optimization. 
Recall that all three algorithms are identical in terms of expressiveness, predicting the same set of races.


\begin{table}[t]
    \caption{Results of data race prediction by \shb, \wcp, \osr, \syncp and $\grconfp[\hSz^5, \hSh, \hLRU^{100}]$. Column 1--2 respectively list the name and number of events ($\mathcal{N}$) for each benchmark. Column 3--12 presents the detected races and processing time by each prediction algorithm. $n(l)$ denotes that $n$ events are predicted as racy with an earlier event and correspond to $l$ bug locations. $n{\color{red} +m}(l{\color{red} +k})$ denotes ${\color{red} m}$ additional races and ${\color{red}k}$ additional bug locations detected against \syncp. OOM denotes an Out-of-Memory (128 GB) error.}
    \tablabel{win}
    \centering
    \scalebox{0.73}{
    \begin{tabular}{|c|c|||c|c||c|c||c|c||c|c||c|c|}
        \hline
        1 & 2 & 3 & 4 & 5 & 6 & 7 & 8 & 9 & 10 & 11 & 12 \\\hline
        \multicolumn{2}{|c|||}{Benchmark} & \multicolumn{2}{c||}{\shb} & \multicolumn{2}{c||}{\wcp} & \multicolumn{2}{c||}{\osr} & \multicolumn{2}{c||}{\syncp} & \multicolumn{2}{c|}{$\optgrconfp$} \\ \hline
        Name & $\mathcal{N}$ & Race & Time & Race & Time & Race & Time & Race & Time & Race & Time \\ \hline
        {\tt  bbuffer} & 9 & 3(1) & 0.2s & 1(1) & 0.16s & 3(1) & 1.20s & 3(1) & 0.01s & 3(1) & 0.01s \\ \hline
        {\tt   array} & 11 & 0(0) & 0.18s & 0(0) & 0.24s & 0(0) & 0.55s & 0(0) & 0.01s & 0(0) & 0.01s \\ \hline
        {\tt   critical} & 11 & 3(3) & 0.36s & 1(1) & 0.08s & 3(3) & 1.23s & 3(3) & 0.01s & 3(3) & 0.01s \\ \hline
        {\tt  account} & 15 & 3(1) & 0.17s & 3(1) & 0.23s & 3(1) & 0.62s & 3(1) & 0.01s & 3(1) & 0.01s \\ \hline
        {\tt  airltkts} & 18 & 8(3) & 0.21s & 5(2) & 0.07s & 8(3) & 0.67s & 8(3) & 0.02s & 8(3) & 0.01s \\ \hline
        {\tt  pingpong} & 24 & 8(3) & 0.19s & 8(3) & 0.09s & 8(3) & 0.66s & 8(3) & 0.01s & 8(3) & 0.01s \\ \hline
        {\tt  twostage} & 83 & 4(1) & 0.15s & 4(1) & 0.16s & 8(2) & 0.68s & 4(1) & 0.05s & 4{\color{red} +4}(1{\color{red} +1}) & 0.74s \\ \hline
        {\tt  wronglock} & 122 & 12(2) & 0.14s & 3(2) & 0.14s & 25(2) & 0.69s & 25(2) & 0.1s & 25(2) & 0.15s \\ \hline
        {\tt  batik} & 131 & 10(2) & 0.2s & 10(2) & 0.14s & 10(2) & 1.01s & 10(2) & 0.02s & 10(2) & 0.07s \\ \hline
        {\tt  mergesort} & 167 & 1(1) & 0.4s & 1(1) & 0.28s & 5(2) & 0.57s & 3(1) & 0.03s & 3{\color{red} +2}(1{\color{red} +1}) & 0.1s \\ \hline
        {\tt  prodcons} & 246 & 1(1) & 0.22s & 1(1) & 0.54s & 1(1) & 0.63s & 1(1) & 0.04s & 1(1) & 0.13s \\ \hline
        {\tt  raytracer} & 526 & 8(4) & 0.14s & 8(4) & 0.27s & 8(4) & 0.62s & 8(4) & 0.03s & 8(4) & 0.12s \\ \hline
        {\tt  biojava} & 852 & 2(2) & 0.2s & 2(2) & 0.16s & 7(3) & 0.66s & 7(3) & 0.05s & 7(3) & 0.36s \\ \hline
        {\tt  clean} & 867 & 59(4) & 0.31s & 33(4) & 0.4s & 110(4) & 0.74s & 60(4) & 0.08s & 60{\color{red} +43}(4) & 0.34s \\ \hline
        {\tt  bubblesort} & 1.65K & 269(5) & 0.2s & 100(5) & 0.33s & 374(5) & 1.66s & 269(5) & 0.27s & 269{\color{red} +98}(5) & 30.97s \\ \hline
        {\tt  lang} & 1.81K & 400(1) & 0.27s & 400(1) & 0.24s & 400(1) & 0.63s & 400(1) & 0.08s & 400(1) & 1.16s \\ \hline
        {\tt  jigsaw} & 3.18K & 4(4) & 0.61s & 4(4) & 0.24s & 6(6) & 0.71s & 6(6) & 0.51s & 6(6) & 16.81s \\ \hline
        {\tt  sunflow} & 3.32K & 84(6) & 0.26s & 58(6) & 0.45s & 130(7) & 0.82s & 119(7) & 1.29s & 119(7) & 1.25s \\ \hline
        {\tt  montecarlo} & 7.60K & 5066(3) & 0.24s & 3267(1) & 0.3s & 5066(3) & 0.75s & 5066(3) & 0.1s & 5066(3) & 0.44s \\ \hline
        {\tt  readwrite} & 9.88K & 92(4) & 0.42s & 92(4) & 0.54s & 228(4) & 0.98s & 199(4) & 0.3s & 199{\color{red} +29}(4) & 38.52s \\ \hline
        {\tt  bufwriter} & 10.26K & 8(4) & 0.66s & 8(4) & 0.49s & 8(4) & 0.87s & 8(4) & 0.27s & 8(4) & 1.31s \\ \hline
        {\tt  luindex} & 15.95K & 1(1) & 0.38s & 2(2) & 0.77s & 15(15) & 0.87s & 15(15) & 0.18s & 15(15) & 1.36s \\ \hline
        {\tt  ftpserver} & 17.10K & 69(21) & 0.51s & 69(21) & 0.85s & 85(21) & 1.05s & 85(21) & 3.92s & 85(21) & 5m9s \\ \hline
        {\tt  moldyn} & 21.07K & 103(3) & 0.39s & 103(3) & 0.7s & 103(3) & 0.90s & 103(3) & 0.19s & 103(3) & 0.71s \\ \hline
        {\tt  derby} & 75.08K & 29(10) & 0.59s & 28(10) & 1.57s & 30(11) & 3.15s & 29(10) & 10.51s & 29{\color{red} +1}(10{\color{red} +1}) & 25m24s \\ \hline
        {\tt  graphchi} & 147.24K & 13(4) & 0.99s & 11(4) & 1.2s & 75(5) & 4.83s & 71(4) & 3.52s & 71{\color{red} +4}(4{\color{red} +1}) & 22.99s \\ \hline
        {\tt  avrora} & 204.11K & 0(0) & 1.24s & 0(0) & 1.64s & 0(0) & 1.22s & 0(0) & 0.37s & 0(0) & 49m37s \\ \hline
        {\tt  hsqldb} & 647.53K & 190(190) & 2.62s & 161(161) & 24.63s & 193(193) & 1m 20s & OOM & - & 191(191) & 22m15s \\ \hline
        {\tt  xalan} & 671.79K & 31(10) & 2.02s & 21(7) & 13.89s & 37(12) & 1m 16s & 37(12) & 1m 47s & 37(12) & 14m1s \\ \hline
        {\tt  lusearch} & 751.32K & 232(44) & 1.73s & 119(27) & 5.36s & 232(44) & 3.09s & 232(44) & 4.26s & 232(44) & 30m29s \\ \hline
        {\tt  lufact} & 891.51K & 21951(3) & 2.1s & 21951(3) & 2.82s & 21951(3) & 2.78s & 21951(3) & 29.23s & 21951(3) & 24.18s \\ \hline
        {\tt  linkedlist} & 910.60K & 5973(4) & 2.47s & 5949(3) & 3.94s & 7095(4) & 4.83s & 7095(4) & 1m 51s & 7095(4) & 2h5m \\ \hline
        {\tt  cryptorsa} & 130.92K & 11(5) & 2.6s & 11(5) & 5.78s & 35(7) & 1m21s & 35(7) & 1m 52s & 35(7) & 25m52s \\ \hline
        {\tt  sor} & 1.90M & 0(0) & 2.92s & 0(0) & 9.07s & 0(0) & 26.10s & 0(0) & 5.27s & 0(0) & 15.48s \\ \hline
        {\tt  tsp} & 15.22M & 143(6) & 24.59s & 140(6) & 47.53s & 143(6) & 1m19s & 143(6) & 1m 15s & 143(6) & 22m54s \\ \hline
        
    \end{tabular}
    
}
\end{table}

\myparagraph{Effectiveness of heuristics}
We now focus on our main algorithm $\grconfp$ and its optimized variant $\optgrconfp$,
evaluating the effectiveness of our heuristics in balancing scalability and expressiveness (RQ3).
The quantile plot in \figref{heuristic} compares the performance 
of $\optgrconfp$ together with variants obtained by applying one or more several heuristics.
For this evaluation, we set $m=5$ for the grain-size heuristic and $n = 100$ for the LRU heuristic.
The comprehensive table that includes the number of detected races and corresponding processing time is provided in 
\iftoggle{appendix}{\appref{full-result}}{the Appendix}.
The results demonstrate that when all heuristics are enabled, $\optgrconfp$ outperforms other configurations in terms of scalability, successfully processing all benchmarks while significantly reducing processing time on ``harder'' benchmarks.

Enabling more heuristics generally improves scalability, however, this scalability comes at a prediction power cost. 
This trade-off is exemplified in two specific cases:
In benchmark \texttt{clean}, 7 extra races are predicted when disabling the ``grain-shape-limitation'' heuristic, suggesting that grains with either partial critical sections or multiple critical sections
play a significant role in this benchmark.
In benchmark \texttt{sunflow}, more than 5 additional races are detected when disabling the ``grain-size'' heuristic, indicating the necessity of larger grain choices.
Notably, despite missing some race events, more aggressive heuristic configurations still identify all unique bug locations. 
This suggests that our heuristics are able to strike a good balance of predictive power
and performance, despite being theoretically incomplete.


\myparagraph{Expressiveness of $\optgrconfp$}
\tabref{win} presents a comprehensive comparison of the expressiveness of granular prefix reasoning and state-of-the-art data race prediction algorithms.
While the complete version of $\grconfp$ does not scale well, 
we evaluate its optimized variant, $\optgrconfp$, 
configured with heuristic parameters $[\hSz^5, \hSh, \hLRU^{100}]$. 
This version maintains scalability across our benchmark 
suite while maintaining a good prediction power (expressiveness).
Our results demonstrate that $\optgrconfp$ outperforms \syncp, 
detecting $181$ more data race events and identifying $4$ more unique bug locations. 
This empirical evidence supports that the increase in theoretical
expressiveness translates to more races found in practice by $\grconfp$ against \syncp; 
in fact, by an optimized version of it 
which is not as expressive as the ideal theoretical algorithm.
Furthermore, although $\grconfp$ is theoretically incomparable with $\osr$ and $\wcp$ in their predictive power,
$\optgrconfp$ identifies all and 81 additional bug locations detected by $\wcp$, falling short by only 2 bug locations compared with $\osr$ in one benchmark ${\tt hsqldb}$.


\begin{wrapfigure}[9]{r}{0.27\textwidth}
    \vspace{-15pt}
    \includegraphics[scale=0.27]{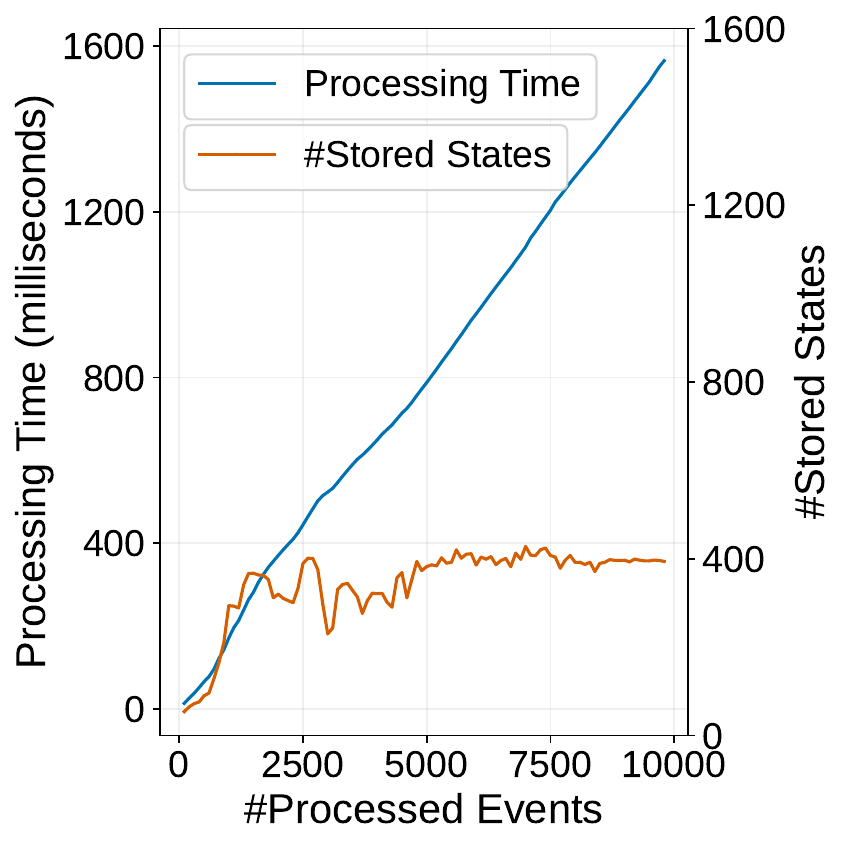}
    \vspace{20pt}
\end{wrapfigure}

\myparagraph{Performance of $\optgrconfp$}
Regarding performance, \osr and $\syncp$ require linear space, while $\grconfp$ is constant-space. 
Nevertheless, \osr and $\syncp$ performs much better in practice,
particularly when processing longer or ``harder'' benchmarks. 
While automata-theoretic algorithms, such as those for \seqp and $\grconfp$, compose well together, 
they suffer from the drawback that the entire state space has to be examined and updated for each incoming event. 
Even for a reasonably small number of states (e.g. 400 states as illustrated on the right), 
repeatedly processing the state space over a very long execution impacts efficiency, potentially limiting scalability.
Future research directions could focus on developing additional heuristics and optimizations for $\grconfp$ and other automata-theoretic algorithms, alongside engineering improvements such as GPU-based parallelization, to enhance their computational efficiency and practical applicability.



\section{Related Work}
\seclabel{related}

Concurrency bug detection techniques have been studied extensively
for multiple decades, with a focus on a large class of bugs, but primarily catered towards data races.
Static analyzers~\cite{Naik06,racerd} typically focus on proving the absence of bugs
in the entire program but tend to raise false alarms, 
limiting their popularity amongst developers.
Dynamic analysis techniques, on the other hand, analyze individual executions
typically suffer from no or very few false positives, and can augment
stress testing~\cite{Musuvathi08} techniques such as fuzz testing~\cite{RFF2024}, controlled concurrency testing~\cite{Nekara2022,yuan2018partial,burckhardt2010randomized}
or model checking~\cite{QR05,Kokologiannakis2022,Kokologiannakis2022}.
The earliest dynamic analysis techniques include those based
on Eraser-style lockset analysis~\cite{Savage97}, which check if
each memory location is protected by a common lock.
Programs that violate this discipline may still be race-free, this algorithm
is again prone to report false positives.

The focus of our work is sound and efficient dynamic analysis.
The simplest sound dynamic analyzers are those that detect a bug \emph{as it occurs}
in runtime by monitoring executions~\cite{kcsan}.
While their runtime overhead is low making them suitable for
performance-critical applications such as the Linux kernel.,
they often miss out on simple data races.
Often the efficacy of such techniques is enhanced using delay injection~\cite{tsvd2019,dataCollider2010}.
In sharp contrast to these are predictive techniques, which attempt to infer the presence of
data races in alternative executions characterized by
a causal model~\cite{Koushik05,serbanuta2013maximal}.
In the simplest form, the happens-before partial order~\cite{Lamport78} 
has inspired some of the most popular and fast dynamic analyses~\cite{Pozniansky03} 
that can be implemented using timestamping~\cite{Mattern89,Fidge91,treeClocks2022}
or lockset-style algorithms~\cite{Elmas07}.
Modern industrial-strength race detectors~\cite{threadsanitizer} are primarily based 
on the epoch optimization~\cite{Flanagan09} for happens-before based race detection.

In its most general form, data race prediction is an 
intractable problem~\cite{Mathur2020b}.
Early algorithms for data race prediction relied upon 
either explicit~\cite{Koushik05} or symbolic~\cite{Said11,Wang09,rvpredict}
enumeration of an exponentially large space of interleavings, and struggled
to scale beyond short executions, unlike happens-before style race detectors.
Smaragdakis et al~\cite{Smaragdakis12} proposed the \emph{causally-precedes}
partial order, and an polynomial time algorithm
based on it for data race prediction.
Kini et al~\cite{Kini17} proposed the first linear time sound dynamic
data race prediction algorithm  based on the \emph{weak causally precedes} relation.
These techniques take inspiration from those based on the happens-before partial order,
and infer orderings between events that potentially reflect causality,
and declare conflicting events to be racy when they are not ordered.
Subsequently, many more partial order based algorithms
have been proposed which either ensure soundness by design~\cite{genc2019}
or using an additional post-hoc analysis~\cite{Roemer20,Roemer18}.
Algorithms such as M2~\cite{Pavlogiannis2020}, $\syncp$~\cite{Mathur21}
and \osr~\cite{OSR2024} are based on characterizations of prefixes that can witness races.

In our work, we demystify and formalize the connection 
between the reasoning beyond these prior polynomial time algorithms and 
reasoning based on commutativity and prefixes, with the goal of arriving at efficient
and more predictive algorithms.
The first starting point is the trace-theoretic interpretation of happens-before, which
lends itself to automata-theoretic techniques, 
naturally giving rise to streaming, constant space and linear time algorithms.
This observation was extended in~\cite{FarzanMathurPOPL2024} for arriving at
more predictive algorithms for causal concurrency.
The automata-theoretic connections to simple prefix reasoning~\cite{Mathur18,Mathur21} 
were formalized in~\cite{Ang2024CAV}.
Here, we show that, as such, vanilla prefix reasoning can subsume reasoning purely based on commutativity
in the context of data race prediction.
We expect the same results to hold for deadlock prediction~\cite{Tunc2023deadlock} 
which asks to solve a similar problem. 
We then outline how careful but intuitive combinations can enhance predictive power
without sacrificing asymptotic complexity of constant space.
We remark that, as such though, partial order based sound techniques of~\cite{Smaragdakis12,Kini17} 
and their derivatives~\cite{genc2019}, as well as
those that attempt to construct a one-off linearization~\cite{OSR2024,Pavlogiannis2020}
are all incomparable to the class of
$\grconfp$-races we propose here.
See \iftoggle{appendix}{\appref{theoretical-comparison}}{the Appendix} for a theoretical comparison.


\section{Conclusions}
\seclabel{conclusions}


Commutativity reasoning has been instrumental in tractable defining equivalences, 
and its applications like model checking~\cite{Flanagan2005}.
Prefix reasoning, on the other hand, as emerged largely in the setting of dynamic bug prediction
against safety properties, and carefully leverage the seemingly
orthogonal axis of carefully choosing which subset of events
must participate in the witness to the bug~\cite{Mathur21,Ang2024CAV},
and, as we show, prefix reasoning subsumes commutativity reasoning for the case of 
data race prediction. In this paper, we combine the two paradigms in a modular manner, a choice of prefix followed by a choice of a pair of grains within which commutativity reasoning is used to predict a race that is beyond the reach of either method individually. The fact that both types of reasoning lend themselves to automata-theoretic constant-space linear-time algorithms enables a modular algorithmic solution for this modular definition.
It is unclear whether similar modular solutions can be devised for combining other algorithm that use (super-)linear space.
It also would be interesting to investigate whether prefix reasoning alone, or combinations of it with commutativity reasoning can play a role in other contexts (such as model checking or proof simplification \cite{lics2023}) where commutativity reasoning alone has had a long history of huge impact. 

\section*{Data Availability Statement}

We have released an anonymized version of our tool $\grconfp$~\cite{grconfp}.
We will submit the tool for artifact evaluation. 

\bibliographystyle{ACM-Reference-Format}
\bibliography{references}

\newpage
\iftoggle{appendix}{
  \appendix


\section{Proofs from \secref{commutativity-prefixes}}

\commutativityComparePredictivePower*

\begin{proof}[Proof Sketch]
It follows simply from the soundness of the independence relations
(in particular that $\scgraincl{\tr} \subseteq \rfcl{\tr} \subseteq \creorderings{\tr}$), and from 
the observations that for every execution $\tr$,
we have $\mazcl{\tr} \subseteq \graincl{\tr} \subseteq \scgraincl{\tr}$
as proved in~\cite{FarzanMathurPOPL2024}.
Here, $\rfcl{\tr}$ is the set of executions that are reads-from equivalent to $\tr$,
and precisely coincide with the set of those correct reorderings $\rho$ of $\tr$
for which $\events{\tr} = \events{\rho}$.
\end{proof}

\commutativityRaceConstantSpace*

\begin{proof}[Proof Sketch]
The proof of the case of $C = \maz$ follows from folklore results,
also reviewed in~\cite{FarzanMathurPOPL2024} since causal concurrency and $C$-race
coincide in this case.
We will present the proof of $C = \scatteredgrains$ and the proof of $C = \grains$ is
similar and skipped.

W.l.o.g assume that $e_1 \trord{\tr} e_2$.
When $e_1$ and $e_2$ are an $\scatteredgrains$-race,
then they can be detected by a choice of grains $S = \set{g^{(1)}, \ldots, g^{(k)}}$
such that $e_1 \in g_1$ and $e_2 \in g_2$.
Let $S_1$ and $S_2$ be the strongly connected components containing
$g_1$ and $g_2$ respectively, in the grain graph of $S$.
If $S_1 = S_2$, then the two grains must not have any other grain
that is sandwiched between them in the middle (according to the edge relation), and further,
$e_1$ must be the last event of $g_1$ and $e_2$ must be the first event of $g_2$,
since the only reorderings we consider are those that linearize $S_1 = S_2$.
If $S_1 \neq S_2$, then we can have three cases.
Case-1: $S_1$ has a path to $S_2$. In this case again, we want that $e_1$ is the last event in $\bigcup_{g \in S_1} \events{g}$, and that $e_2$ is the first event in $\bigcup_{g \in S_2} \events{g}$. Finally, we also want that no other SCC lies on the path from $S_1$ to $S_2$.
Case-2: $S_2$ has a path to $S_1$. In this case, we want that $e_1$ is the first event in $\bigcup_{g \in S_1} \events{g}$, and that $e_2$ is the last event in $\bigcup_{g \in S_2} \events{g}$. Finally, we also want that no other SCC lies on the path from $S_2$ to $S_1$.
Case-3: No path between $S_1$ and $S_2$. In this case, we either demand that $e_1$ is the first event of $S_1$ and $e_2$ is the last event of $S_2$, or vice versa.
All these conditions can be checked in constant space using the grain graph concurrency monitor of~\cite{FarzanMathurPOPL2024}.
\end{proof}

\seqpRaceIsSyncpRace*

\begin{proof}
The proof of why each \syncp-race is a \confp-race was presented in~\cite{Ang2024CAV}.
Here we show that each \confp-race is a \seqp-race (the other directions are more straightforward).
Suppose $(e_1, e_2)$ is a \confp-race of execution $\tr$.
Then there is a correct reordering $\rho$ of $\tr$ such that $\rho \mazeq \rho'$ and $e_1$ and $e_2$ are $\tr$-enabled in $\tr$,
where $\rho' = \proj{\tr}{\events{\rho}}$.
Consider $\rho'$.
Observe that $\rho'$ is a correct reordering of $\rho$ and thus also a correct
reordering of $\tr$.
Further, each event enabled in $\rho$ is also enabled in $\rho'$.
Finally, $\rho'$ preserves the order of events as in $\tr$.
Thus, $\rho'$ is a \seqp-prefix and thus $(e_1, e_2)$ is a \seqp-race
witnessed by the \seqp-prefix $\rho'$.
\end{proof}

\seqpRaceConstantSpace*

\begin{proof}
Follows from \propref{seqp-prefix-race-is-syncp-race} and the constant space
algorithm of \confp-races~\cite{Ang2024CAV}.
\end{proof}

\subsection{Proof of \thmref{prefix-subsumes-commutativity}}

Let us now prove \thmref{prefix-subsumes-commutativity}.
Towards this, let us first consider a simple result.



\PrefixSubsumesCommutativity*

\begin{proof}
Let us first prove that every $C$-race is also a $P$-race.
For this, it suffices to show that each $\scatteredgrains$-race is also a \seqp-race.
Let $(e_1, e_2)$ be a $\scatteredgrains$-race of $\tr$.
This means that there is an execution $\rho \in \scgraincl{\tr}$ such that
$e_1$ and $e_2$ are consecutive in $\rho$.
This means that there is a choice of scattered grains $S = \set{g^{(1)}, g^{(2)}, \ldots g^{(k)}}$
such that $\rho$ is obtained as a sequence $\rho = \textsf{lin}(\scc_1) \textsf{lin}(\scc_2) \cdots \textsf{lin}(\scc_m)$,
where, $\set{\scc_i}_i$ is the class of SCCs in
the graph $\mathsf{GGraph}_{\tr, S}$.
We will argue that there is another linearization of these SCCs whose prefix is \seqp and it nevertheless still

W.l.o.g assume $e_1 \trord{\tr} e_2$.
Let $\scc_1$ and $\scc_2$ be the SCCs of $e_1$ and $e_2$ respectively.
It is of course possible that $\scc_1 = \scc_2$, in which case no event that is between $e_1$
and $e_2$ in $\trord{\tr}$ occurs in this $\scc$.
Now, consider $S' = \setpred{\scc}{\scc \text{ is an SCC of } \mathsf{GGraph}_{\tr, S}, \scc \rightsquigarrow
\scc_1 \lor \scc \rightsquigarrow \scc_2 } \setminus \set{\scc_1, \scc_2}$, 
where $\scc \rightsquigarrow \scc'$ denotes that there is a
path in the graph from some grain in $\scc$ to some grain in $\scc'$.
Now, consider the following set of events
\begin{align*}
E = 
\underbrace{\cup_{g \in S'} \events{g}}_{E_{\sf down}}
\cup 
\underbrace{\setpred{e \neq e_1}{e \in \scc_1, e \trord{\tr} e_1 }}_{E_1}
\end{align*}
We will now show that 
(1) $E$ is downward closed with respect to $\po{}$ and $\rf{}$, 
(2) $E$ does not have multiple open acquire events for the same lock,
(3) if any acquire is unmatched in $E$, then it is the last acquire (according to $\trord{\tr}$) on that lock, and finally
(4) the $\po{}$-predecessors of both $e_1$ and $e_2$.
(1), (2) and (3) ensure that when $E$ is linearized according to $\trord{\tr}$,
we do get a well-formed execution, which means it is also a \seqp-prefix of $\tr$, whereas
(4) ensures that $e_1$ and $e_2$ are $\tr$-enabled after this prefix.

(1) is easy to argue for events in $E_{\sf down}$. Now consider an event $e \in E_1$.
And let $f$ be such that $(f, e) \in \po{\tr} \cup \rf{\tr}$.
It is easy to see that if $f$ is not in $\scc_1$, then it must be in an SCC from $S'$,
in which case it will be in $E_{\sf down}$.
Thus, (1) is true even for $E_1$.

Let us now prove (2) and (3) together.
Suppose on the contrary that there is a lock $\lk$ and there are two acquires $a_1 \stricttrord{\tr} a_2$
of $\lk$ in $E$, both whose matching releases (resp. $r_1$ and $r_2$) are not in $E$.
In this case the grains of $a_1$ and $a_2$ will not be the same, i.e.,
$g(a_1) \neq g(r_1)$.
But then, in this case, there will be an edge from $g(r_1)$ to $g(a_2)$.
Thus, $g(r_1)$ must also be in $S'$ and thus $r_1 \in E$.
For the same reason, only the last acquire on some lock can be open.

Let us now prove (4). If the $\po{}$-predecessor of $e_1$ is in $\scc_1$,
then it is in the set $E_1$, and otherwise it is in $E_{\sf down}$.
Now, let us argue this about $e_2$.
Since $e_1$ and $e_2$ are next to each other.
If they are both in $\scc_1$, then the argument is similar to $e_1$ (of course we know that $(e_1, e_2) \not\in \po{\tr}$).
Otherwise, we know that $e_2$ must be the first event in ${\sf lin}(\scc_2)$ and $e_1$ must be the last event of ${\sf lin}(\scc_1)$.
In this case, all $\po{}$-predecessors of $e_2$ must be in other SCCs.
If they are SCCs from $S'$, then we are done.
Otherwise, it is $\scc_1$, each of its elements (apart from $e_1$) are also in $E$, and we are done again!

For the other direction, consider the execution in \figref{tr-syncp-not-comm}.
\begin{figure}[h]
\subfloat[Execution $\trsyncpnotcomm$ \figlabel{tr-syncp-not-comm}]{%
	\execution{2}{
		\figev{2}{$\wt(y)$}
		\figev{1}{\racy{$\wt(x)$}} 
		\figev{1}{$\acq(\lk)$}
		\figev{1}{$\rd(y)$}
		\figev{1}{$\rel(\lk)$}
		\figev{2}{$\acq(\lk)$}
		\figev{2}{$\wt(y)$}
		\figev{2}{$\rel(\lk)$}
		\figev{2}{\racy{$\wt(x)$}}
	}
}
\subfloat[Execution $\trsyncpnotcommreordering$ \figlabel{tr-syncp-not-comm}]{%
    \execution{2}{
		\figev{2}{$\wt(y)$}
		\figev{2}{$\acq(\lk)$}
		\figev{2}{$\wt(y)$}
		\figev{2}{$\rel(\lk)$}
		\figev{1}{\racy{$\wt(x)$}} 
		\figev{2}{\racy{$\wt(x)$}}
	}
}
\figlabel{syncp-race-not-commutativity-race}
\end{figure}
Here, the two write events on $x$ are a \seqp-race, as witnessed
by the \seqp-prefix shown in \figlabel{tr-syncp-not-comm}.
However, this race cannot be witnessed by any reasoning that preserves the reads-from equivalence
on the set of all events, because in any correct reordering that
witnesses this race, the event $\ev{T_1, \rd(y)}$ cannot be present.
This rules out the fact that there is an execution $\rho \in \scgraincl{\tr}$
in which $e_1$ and $e_2$ are consecutive.
\end{proof}

\section{Proofs from \secref{combining}}

\section{Comparing $\grconfp$ with other classes of data races}
\applabel{theoretical-comparison}

Here, we compare the class of $\grconfp$ races with that characterized by
prior work.

\subsection{Comparison with \shb and \syncp}

As we discuss in \thmref{grconfp-subsumes-syncp}, every \syncp-race~\cite{Mathur21} 
is a $\grconfp$-race.
Further, since every race detected using the \emph{schedulable-happens-before} (\shb)
partial order~\cite{Mathur18}
is also a \syncp-race, they are also $\grconfp$-races.
Next, from \thmref{prefix-subsumes-commutativity}, it also follows that every $C$-augmented $P$-race is also
a $\grconfp$-race, for every $C \in \set{\maz, \grains, \scatteredgrains}$ and for every
$P \in \set{\seqp, \confp, \syncp}$.
Finally, the execution in \figref{syncp-miss} demonstrates a race
that is a $\grconfp$-race but not a \syncp-race.

\subsection{Comparison with \cp, \wcp, \sdp, \mtwo and \osr}

\begin{figure}[h]
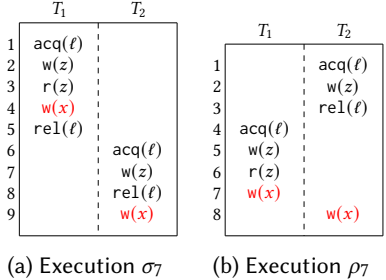

\subfloat[Execution $\trgrconfpwins$ \figlabel{tr-grconfp-wins}]{%
	\execution{2}{
		\figev{1}{$\acq(\lk)$}
		\figev{1}{$\wt(z)$} 
		\figev{1}{$\rd(z)$} 
		\figev{1}{\racy{$\wt(x)$}} 
		\figev{1}{$\rel(\lk)$}
		\figev{2}{$\acq(\lk)$}
		\figev{2}{$\wt(z)$} 
		\figev{2}{$\rel(\lk)$}
		\figev{2}{\racy{$\wt(x)$}}
	}
}
\subfloat[Execution $\trgrconfpwinsreordering$ \figlabel{tr-grconfp-wins-reordering}]{%
    \execution{2}{
    	\figev{2}{$\acq(\lk)$}
		\figev{2}{$\wt(z)$} 
		\figev{2}{$\rel(\lk)$}
		\figev{1}{$\acq(\lk)$}
		\figev{1}{$\wt(z)$} 
		\figev{1}{$\rd(z)$} 
		\figev{1}{\racy{$\wt(x)$}} 
		\figev{2}{\racy{$\wt(x)$}}
	}
}
\caption{Race detected by $\grconfp$ but not by \cp, \wcp, \sdp, \mtwo or \osr. 
The reordering on the right shows the witness reordering obtained by a $\grconfp$-prefix.}
\figlabel{grconfp-race-missed-by-cp-wcp-sdp-m2}
\end{figure}

Smaragdakis et al~\cite{Smaragdakis12} introduced the \emph{causally precedes} (\cp)
partial order that weakens the happens-before order and can be implemented in super-linear time.
Subsequently, Kini et al~\cite{Kini17} generalize it to
the \emph{weak causally precedes} (\wcp) relation which could 
be implemented in linear time and space.
Next, Gen{\c c} et al~\cite{genc2019} further weakened \wcp to the 
\emph{strong-dependently-precedes} (\sdp) order.
The races characterized by these three orders (called \cp-races, \wcp-races and \sdp-races)
are exactly those pairs of conflicting events that are unordered by them.
Since these races form a larger class than \hb-races, these orders essentially identify
conditions under which two critical sections on the same lock do not need to be ordered. 
At the same time, in order to ensure soundness and tractability, they can be conservative
in their judgement.
\figref{tr-grconfp-wins} shows an example of a data race (between the two $\wt(x)$ events in $T_1$ and $T_2$) 
that each of these three partial order based techniques fail to detect.
This is because these relations will order the two critical sections on $\lk$
since the first one contains a $\rd(z)$ event, while the later critical section
contains a conflicting $\wt(z)$ event.
Composition with $\po{}$ then tells us that the two racy events actually get ordered and thus cannot
be declared to be a race.

The \mtwo algorithm~\cite{Pavlogiannis2020} attempts
to prove that an event pair $(e_1, e_2)$ is a race by coming up with a 
prefix that closes every open acquire in the prefix --- if in doing so,
either $e_1$ or $e_2$ are forced into the prefix, this pair is not a race, else it is.
The race in~\figref{tr-grconfp-wins} is not a \mtwo race since
closing the acquire events on $\lk$ will enforce the earlier $\wt(x)$ event to be in the prefix.
\osr-races~\cite{OSR2024} are characterized by prefixes where
locks are only optimistically closed,
but require that the order of conflicting write events is not reversed.
The race in \figref{tr-grconfp-wins} is also not an \osr-race~\cite{OSR2024}
since the only reordering that witnesses it~\figref{tr-grconfp-wins-reordering}
reverses the order of the write events on $z$.

\begin{figure}[h]
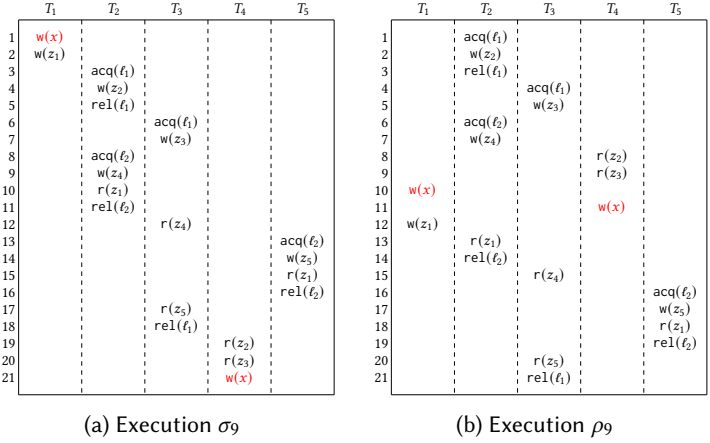

\subfloat[Execution $\trosrmisses$ \figlabel{tr-osr-misses}]{%
\scalebox{0.8}{
	\execution{5}{
	    \figev{1}{\racy{$\wt(x)$}}
	    \figev{1}{$\wt(z_1)$}
	    \figev{2}{$\acq(\lk_1)$}
	    \figev{2}{$\wt(z_2)$}
	    \figev{2}{$\rel(\lk_1)$}
	    \figev{3}{$\acq(\lk_1)$}
	    \figev{3}{$\wt(z_3)$}
	    \figev{2}{$\acq(\lk_2)$}
	    \figev{2}{$\wt(z_4)$}
	    \figev{2}{$\rd(z_1)$}
	    \figev{2}{$\rel(\lk_2)$}
	    \figev{3}{$\rd(z_4)$}
	    \figev{5}{$\acq(\lk_2)$}
	    \figev{5}{$\wt(z_5)$}
	    \figev{5}{$\rd(z_1)$}
	    \figev{5}{$\rel(\lk_2)$}
	    \figev{3}{$\rd(z_5)$}
	    \figev{3}{$\rel(\lk_1)$}
	    \figev{4}{$\rd(z_2)$}
	    \figev{4}{$\rd(z_3)$}
	    \figev{4}{\racy{$\wt(x)$}}
    }
   }
}
\subfloat[Execution $\trosrmissesreordering$ \figlabel{tr-osr-misses-reordering}]{%
  \scalebox{0.8}{
  \execution{5}{
	    \figev{2}{$\acq(\lk_1)$}
	    \figev{2}{$\wt(z_2)$}
	    \figev{2}{$\rel(\lk_1)$}
	    \figev{3}{$\acq(\lk_1)$}
	    \figev{3}{$\wt(z_3)$}
	    \figev{2}{$\acq(\lk_2)$}
	    \figev{2}{$\wt(z_4)$}
	    \figev{4}{$\rd(z_2)$}
	    \figev{4}{$\rd(z_3)$}
	    \figev{1}{\racy{$\wt(x)$}}
	    \figev{4}{\racy{$\wt(x)$}}
	    \figev{1}{$\wt(z_1)$}
	    \figev{2}{$\rd(z_1)$}
	    \figev{2}{$\rel(\lk_2)$}
	    \figev{3}{$\rd(z_4)$}
	    \figev{5}{$\acq(\lk_2)$}
	    \figev{5}{$\wt(z_5)$}
	    \figev{5}{$\rd(z_1)$}
	    \figev{5}{$\rel(\lk_2)$}
	    \figev{3}{$\rd(z_5)$}
	    \figev{3}{$\rel(\lk_1)$}
    }
    }
} 
\caption{Race missed by \osr and \mtwo but not by $\maz$. 
The reordering on the right shows the witness reordering obtained by a \seqp-prefix.}
\figlabel{osr-misses-hb}
\end{figure}
It is noteworthy that \osr can miss even $\maz$-races~\cite{Mathur18};
see \figref{osr-misses-hb}

On the other hand, each of these methods can detect races that require
linear space~\cite{Kini17} or quadratic time~\cite{OSR2024} 
and thus cannot be $\grconfp$-races since the latter
can be detected in constant space and linear time.

\section{Proofs from \secref{stratification}}

\Monotonicity*

\begin{proof}
    Since $e_1, e_2$ is a predictable race in $\beta$, there exists a correct reordering $\rho$ of $\beta$ where $e_1, e_2$ are enabled by $\rho$.
    To show $e_1, e_2$ also form a predictable race in $\sigma$, it suffices to prove that $\alpha\rho$ is a correct reordering of $\sigma$.
    We have $\events{\alpha\rho} \subseteq\events{\alpha\beta}$ as $\events{\rho}\subseteq \events{\beta}$.
    For each $e\in\threads$, $\proj{\alpha\rho}{t} = \proj{\alpha}{t}\proj{\rho}{t}$ is a prefix of $\proj{\alpha}{t}\proj{\beta}{t} = \proj{\sigma}{t}$.
    For reads-from relation, $\rf{\alpha\rho} = \rf{\alpha} \cup \rf{\rho} \cup \setpred{e\in \rho, e'\in\alpha}{(e, e')\in \rf{\alpha\beta}} \subseteq \rf{\alpha} \cup \rf{beta} \cup \setpred{e\in \beta, e'\in\alpha}{(e, e')\in \rf{\alpha\beta}} = \rf{\alpha}\rf{\beta}$.
\end{proof}

\section{Complete Table for \secref{experiments}}
\applabel{full-result}
\begin{table}[!ht]
    \caption{Results of data race prediction by $\optgrconfp$ under different combinations of heuristics$[-/\hSz^5, -/\hSh, -/\hLRU^{100}]$. Column 1--2 respectively list the name and number of events ($\mathcal{N}$) for each benchmark. Column 3--18 presents the detected races and processing time by each heuristic configurations. $n(l)$ denotes that $n$ events are predicted as racy with an earlier event and correspond to $l$ bug locations. TO denotes a time-out on 3 hours.}
    \tablabel{heuristics}
    \centering
    \scalebox{0.50}{
    \begin{tabular}{|c|c|||c|c||c|c||c|c||c|c||c|c||c|c||c|c||c|c|}
        \hline
        1 & 2 & 3 & 4 & 5 & 6 & 7 & 8 & 9 & 10 & 11 & 12 & 13 & 14 & 15 & 16 & 17 & 18 \\\hline
        \multicolumn{2}{|c|||}{Benchmark} & \multicolumn{2}{c||}{\makecell{[$\hSh$, $\hSz^5$, $\hLRU^{100}$]}}  & \multicolumn{2}{c||}{\makecell{[$-$, $\hSz^5$, $\hLRU^{100}$]}} & \multicolumn{2}{c||}{\makecell{[$\hSh$, $-$, $\hLRU^{100}$]}} & \multicolumn{2}{c||}{\makecell{[$\hSh$, $\hSz^5$, $-$]}} & \multicolumn{2}{c||}{\makecell{[$-$, $-$, $\hLRU^{100}$]}} & \multicolumn{2}{c||}{\makecell{[$-$, $\hSz^5$, $-$]}} & \multicolumn{2}{c||}{\makecell{[$\hSh$, $-$, $-$]}} & \multicolumn{2}{c|}{\makecell{[$-$, $-$, $-$]}}\\\hline
        Name & $\mathcal{N}$ & Race & Time & Race & Time & Race & Time & Race & Time & Race & Time & Race & Time & Race & Time & Race & Time \\ \hline
        
        {\tt bbuffer} & 9 & 3(1) & 0.01s & 3(1) & 0.01s & 3(1) & 0.01s & 3(1) & 0.01s & 3(1) & 0.01s & 3(1) & 0.01s & 3(1) & 0.01s & 3(1) & 0.01s \\ \hline
        {\tt  array} & 11 & 0(0) & 0.01s & 0(0) & 0.02s & 0(0) & 0.01s & 0(0) & 0.01s & 0(0) & 0.02s & 0(0) & 0.02s & 0(0) & 0.01s & 0(0) & 0.03s \\ \hline
        {\tt  critical} & 11 & 3(3) & 0.01s & 3(3) & 0.01s & 3(3) & 0.01s & 3(3) & 0.01s & 3(3) & 0.01s & 3(3) & 0.01s & 3(3) & 0.01s & 3(3) & 0.01s \\ \hline
        {\tt account} & 15 & 3(1) & 0.01s & 3(1) & 0.02s & 3(1) & 0.01s & 3(1) & 0.01s & 3(1) & 0.02s & 3(1) & 0.01s & 3(1) & 0.01s & 3(1) & 0.03s \\ \hline
        {\tt airltkts} & 18 & 8(3) & 0.01s & 8(3) & 0.03s & 8(3) & 0.01s & 8(3) & 0.01s & 8(3) & 0.02s & 8(3) & 0.02s & 8(3) & 0.01s & 8(3) & 0.03s \\ \hline
        {\tt pingpong} & 24 & 8(3) & 0.01s & 8(3) & 0.03s & 8(3) & 0.01s & 8(3) & 0.01s & 8(3) & 0.03s & 8(3) & 0.03s & 8(3) & 0.03s & 8(3) & 0.04s \\ \hline
        {\tt twostage} & 83 & 8(2) & 0.74s & 7(2) & 0.25s & 8(2) & 0.15s & 8(2) & 1m51s & 4(1) & 0.52s & 8(2) & 55m28s & 8(2) & 1m49s & $-$ & TO \\ \hline
        {\tt wronglock} & 122 & 25(2) & 0.15s & 20(2) & 0.25s & 25(2) & 0.15s & 25(2) & 0.15s & 6(2) & 0.62s & 25(2) & 0.58s & 25(2) & 0.15s & $-$ & TO \\ \hline
        {\tt batik} & 131 & 10(2) & 0.07s & 10(2) & 0.09s & 10(2) & 0.11s & 10(2) & 0.07s & 10(2) & 0.4s & 10(2) & 0.11s & 10(2) & 0.09s & 10(2) & 0.41s \\ \hline
        {\tt mergesort} & 167 & 5(2) & 0.1s & 5(2) & 0.25s & 5(2) & 0.23s & 5(2) & 0.1s & 1(1) & 0.81s & 5(2) & 0.25s & 5(2) & 1.54s & 5(2) & 32m21s \\ \hline
        {\tt prodcons} & 246 & 1(1) & 0.13s & 1(1) & 0.54s & 1(1) & 0.18s & 1(1) & 0.14s & 1(1) & 1.31s & 1(1) & 0.69s & 1(1) & 0.19s & $-$ & TO \\ \hline
        {\tt raytracer} & 526 & 8(4) & 0.12s & 8(4) & 0.16s & 8(4) & 0.23s & 8(4) & 0.13s & 5(3) & 0.98s & 8(4) & 0.23s & 8(4) & 0.21s & 8(4) & 2.68s \\ \hline
        {\tt biojava} & 852 & 7(3) & 0.36s & 7(3) & 1s & 7(3) & 3.21s & 7(3) & 0.36s & 6(3) & 6.83s & 7(3) & 1.43s & 7(3) & 45.94s & $-$ & TO \\ \hline
        {\tt clean} & 867 & 103(4) & 0.34s & 110(4) & 1.35s & 103(4) & 0.43s & 103(4) & 0.35s & 87(4) & 4.62s & 110(4) & 3.84s & 103(4) & 0.41s & $-$ & TO \\ \hline
        {\tt bubblesort} & 1.65K & 367(5) & 30.97s & 322(5) & 7.44s & 274(5) & 9.88s & 371(5) & 2m22s & 44(5) & 1m20s & 371(5) & 5m40s & $-$ & TO & $-$ & TO \\ \hline
        {\tt lang} & 1.81K & 400(1) & 1.16s & 310(1) & 2.66s & 102(1) & 1.85s & 400(1) & 1.12s & 0(0) & 2m48s & 400(1) & 2.68s & 400(1) & 9.12s & 400(1) & 42m57s \\ \hline
        {\tt jigsaw} & 3.18K & 6(6) & 16.81s & 5(5) & 9.12s & 6(6) & 11.97s & 6(6) & 1m11s & 4(4) & 1m21s & 6(6) & 9m17s & 6(6) & 1h22m & $-$ & TO \\ \hline
        {\tt sunflow} & 3.32K & 119(7) & 1.25s & 116(7) & 5.71s & 124(7) & 1.98s & 119(7) & 1.27s & 30(5) & 1m6s & 119(7) & 38m26s & 130(7) & 43m17s & $-$ & TO \\ \hline
        {\tt montecarlo} & 7.60K & 5066(3) & 0.44s & 5066(3) & 0.94s & 5066(3) & 0.49s & 5066(3) & 0.46s & 60(3) & 3.71s & 5066(3) & 0.94s & 5066(3) & 0.43s & $-$ & TO \\ \hline
        {\tt readwrite} & 9.88K & 228(4) & 38.52s & 228(4) & 14.52s & 228(4) & 8.26s & 228(4) & 37.8s & 190(4) & 1m34s & 228(4) & 12m52s & 228(4) & 1m2s & $-$ & TO \\ \hline
        {\tt bufwriter} & 10.26K & 8(4) & 1.31s & 8(4) & 12.15s & 8(4) & 18.24s & 8(4) & 1.33s & 8(4) & 1m1s & 8(4) & 22.64s & 8(4) & 23m9s & $-$ & TO \\ \hline
        {\tt luindex} & 15.95K & 15(15) & 1.36s & 15(15) & 10.43s & 15(15) & 48.93s & 15(15) & 1.37s & 7(7) & 1m21s & 15(15) & 11.46s & $-$ & TO & $-$ & TO \\ \hline
        {\tt ftpserver} & 17.10K & 85(21) & 5m9s & 77(21) & 51.8s & 78(21) & 32.73s & $-$ & TO & 65(21) & 10m11s & $-$ & TO & $-$ & TO & $-$ & TO \\ \hline
        {\tt moldyn} & 21.07K & 103(3) & 0.71s & 103(3) & 1.68s & 103(3) & 3.19s & 103(3) & 0.64s & 80(3) & 19.82s & 103(3) & 1.68s & 103(3) & 19m21s & $-$ & TO \\ \hline
        {\tt derby} & 75.08K & 30(11) & 25m24s & 29(10) & 8m59s & 30(11) & 16m19s & $-$ & TO & 26(9) & 1h8m & $-$ & TO & $-$ & TO & $-$ & TO \\ \hline
        {\tt graphchi} & 147.24K & 0(0) & 22.99s & 0(0) & 3m1s & 0(0) & 22.45s & 0(0) & 22.54s & 0(0) & 30m2s & 0(0) & 5m4s & 0(0) & 24.84s & $-$ & TO \\ \hline
        {\tt avrora} & 204.11K & 75(5) & 49m37s & 70(5) & 6m48s & 74(5) & 8m4s & 75(5) & 1h22m & 16(4) & 1m38m & $-$ & TO & $-$ & TO & $-$ & TO \\ \hline
        {\tt hsqldb} & 647.53K & 191(191) & 22m15s & 176(176) & 44m53s & 173(173) & 1h46m & $-$ & TO & $-$ & TO & $-$ & TO & $-$ & TO & $-$ & TO \\ \hline
        {\tt xalan} & 671.79K & 37(12) & 14m1s & 37(12) & 30m28s & 37(12) & 34m46s & $-$ & TO & $-$ & TO & $-$ & TO & $-$ & TO & $-$ & TO \\ \hline
        {\tt lusearch} & 751.32K & 232(44) & 30m29s & 230(43) & 22m58s & 212(42) & 22m20s & 232(44) & 30m28s & $-$ & TO & $-$ & TO & $-$ & TO & $-$ & TO \\ \hline
        {\tt lufact} & 891.51K & 21951(3) & 24.18s & 21951(3) & 1m42s & 21951(3) & 46.93s & 21951(3) & 25.33s & 20987(3) & 24m20s & 21951(3) & 1m42s & 21951(3) & 7m15s & $-$ & TO \\ \hline
        {\tt linkedlist} & 910.60K & 7095(4) & 2h5m & $-$ & TO & $-$ & TO & $-$ & TO & $-$ & TO & $-$ & TO & $-$ & TO & $-$ & TO \\ \hline
        {\tt cryptorsa} & 130.92K & 35(7) & 25m52s & 30(6) & 1h13m & 31(7) & 1h22m & $-$ & TO & $-$ & TO & $-$ & TO & $-$ & TO & $-$ & TO \\ \hline
        {\tt sor} & 1.90M & 0(0) & 15.48s & 0(0) & 37.14s & 0(0) & 10m20s & 0(0) & 14.8s & 0(0) & 50m32s & 0(0) & 37.67s & 0(0) & 1h18m & $-$ & TO \\ \hline
        {\tt tsp} & 15.22M & 143(6) & 22m54s & $-$ & TO & 143(6) & 19m49s & 143(6) & 20m21s & $-$ & TO & $-$ & TO & 143(6) & 19m55s & $-$ & TO \\ \hline
    \end{tabular}
    }
\end{table}
}

\end{document}